\documentclass[journal,twoside,web]{ieeecolor}
\usepackage{booktabs}
\usepackage{tabularx}
\usepackage{generic}
\usepackage{cite}
\usepackage{amsmath,amssymb,amsfonts}
\usepackage{algorithmic}
\usepackage{graphicx}
\usepackage{textcomp}
\usepackage[caption=false,font=footnotesize]{subfig}
\def\BibTeX{{\rm B\kern-.05em{\sc i\kern-.025em b}\kern-.08em
    T\kern-.1667em\lower.7ex\hbox{E}\kern-.125emX}}
\newtheorem{thm}{Theorem}
\newtheorem{prop}{Proposition}
\newtheorem{lem}{Lemma}
\newtheorem{assum}{Assumption}

\newtheorem{rem}{Remark}
\newtheorem{myproperty}{Property}
\newtheorem{cor}{Corollary}

\makeatletter
\def\@opargbegintheorem#1#2#3{%
  \@IEEEtmpitemindent\itemindent\relax
  \topsep 0pt
  \rmfamily
  \trivlist
  \item[]\textit{\indent #1\ #2\ (#3):} %
  \itemindent\@IEEEtmpitemindent\relax
}
\makeatother

\definecolor{darkred}{RGB}{192, 0, 0}    
\definecolor{darkblue}{RGB}{0, 112, 192}   

\usepackage[normalem]{ulem}

\begin{document}
\title{Bearing-Only Formation Tracking Control for Euler--Lagrange Multi-Agent Systems Without Inter-Agent Communication}
\author{Zilong Song, \IEEEmembership{Graduate Student Member, IEEE}, 
Lu Liu, \IEEEmembership{Senior Member, IEEE}, \\
and Gang Feng, \IEEEmembership{Fellow, IEEE}
\thanks{This work was supported by the Research Grants Council of Hong Kong under Grant 11205024, Grant 11206625, and Grant 11209026. \textit{(Corresponding author: Gang Feng.)}}
\thanks{All authors are with the Department of Mechanical Engineering, City University of Hong Kong, Kowloon, Hong Kong (email: zl.song@my.cityu.edu.hk; lu.liu@cityu.edu.hk; megfeng@cityu.edu.hk).}}

\maketitle

\begin{abstract}
This paper investigates communication-free bearing-only formation tracking control for multi-agent systems governed by Euler--Lagrange dynamics. 
Distinct from existing results that can only stabilize a stationary formation, 
this work considers a scenario where the leaders move with time-varying velocities while the inter-agent communication is absent. 
In this setup, the leaders' states (position and velocity) are unavailable to all followers and cannot be estimated via distributed observers. 
A novel adaptive distributed control scheme is developed to address this problem. 
The design exploits the fact that bearing rates contain the projected relative-velocity information, 
which, together with bearing rigidity, provides a rigidity-based damping mechanism for compensating the unavailable velocity error. 
Moreover, this damping mechanism is incorporated into a bearing-driven auxiliary variable to construct a surrogate velocity error, facilitating the adaptive control design for EL dynamics. 
Furthermore, since this damping mechanism necessitates sufficient bearing rigidity, 
we characterize a rigidity-preserving set and establish its forward invariance, thereby guaranteeing such rigidity via initial conditions.
Via a Filippov-based Lyapunov analysis, the proposed scheme is shown to achieve local practical formation tracking in the sense that the velocity error converges to zero and the position error is uniformly ultimately bounded. 
As a corollary, for the constant-velocity case, asymptotic tracking is achieved without initial-condition restriction.
The simulation results verify the effectiveness of the proposed control law.  
\end{abstract}

\begin{IEEEkeywords}
Bearing-only formation, Bearing rigidity, Euler--Lagrange systems, Formation control, Multi-agent systems.
\end{IEEEkeywords}

\section{Introduction}
\label{sec:introduction}
\IEEEPARstart{F}{ormation} control of multi-agent systems (MASs) has attracted considerable attention owing to its broad applications.
Typically, existing formation control schemes can be classified into position-, displacement-, distance-, and bearing-based ones \cite{OH_Auto_2015}. 
Among them, bearing-based formation control is attractive because inter-agent bearings are directional measurements (line-of-sight directions) that can be easily obtained by vision sensors or antenna arrays \cite{Su_ARC_2026,Dou_IJRNC_2025}. 
In particular, the bearing-only strategy, which employs only bearing measurements, without requiring communication or relative distances, provides a promising solution for formation control in communication-free environments. 

The theoretical foundation for bearing-only formation control is the bearing rigidity \cite{Zhaoshiyu_TAC_2016}, which characterizes when a target formation is uniquely determined, up to a translation and a scaling factor, by inter-agent bearings. 
Together with the bearing localizability analysis \cite{ZhaoShiyu_Auto_2016}, it specifies when the followers' positions are determined from the bearing constraints and the leaders' positions. 
Built upon these foundations, various bearing-only control laws have been developed for formation stabilization.
Notably, almost global convergence to stationary target formations was established for single-integrator agents in \cite{Zhaoshiyu_TAC_2016}, with subsequent extensions to finite-time convergence \cite{TrinhMukherjee_CDC_2017}, directed sensing topologies \cite{TrinhZhaoSun_TAC_2019}, and global orientation estimation \cite{Tran_TCNS_2019}.
Despite these advancements, the aforementioned results are confined to stationary formations and hence inapplicable to scenarios where the formation is required to maneuver as a whole. 

To accommodate such formation requirements, bearing-only formation tracking, where moving leaders induce a time-varying target formation, has been subsequently investigated. 
For leaders moving with a constant velocity, bearing-only tracking laws were developed for single- and double-integrator agents \cite{ZhaoShiyu_TAC_2019}. This framework was subsequently extended to double integrators subject to external disturbances \cite{Trinh_Auto_2021} and local frame alignment \cite{Zhaojianing_LCSS_2021,Garanayak_TCNS_2025}.
More recently, our previous work \cite{Song_LCSS_2024,Song_TIE_2026} considered disturbed single- and double-integrator systems with leaders of time-varying  velocity.
However, these designs rely on the strong assumption that the closed-loop formation maintains sufficient bearing rigidity throughout the evolution, thereby limiting their practical applicability.
Moreover, a common limitation of the aforementioned works is that the agents are modeled as low-order linear integrators. 
Such idealized models, however, fail to capture the inherent nonlinearities, dynamic couplings, and parametric uncertainties of physical platforms, including quadrotors, wheeled mobile robots, and spacecraft, whose motions are more accurately described by Euler--Lagrange (EL) dynamics \cite{ORTEGA_Auto_1989}. 

Formation control of EL systems has been extensively studied in scenarios where relative positions and velocities are measurable or inter-agent communication is available.
In these settings, parametric uncertainties can be effectively handled by exploiting the linear-in-parameter property through adaptive schemes \cite{Meng_TAC_2023_EL-Dym}.
By contrast, relevant results within the bearing-only setting remain rather limited. 
Although the existing works \cite{ZhaobangweiZhaojianingYUxiao_IFAC_2023}, \cite{Li-Song-Xie_TCNS_2025}, \cite{Li-Song_IJRNC_2025} achieve bearing-only formation stabilization for EL systems, they are strictly restricted to stationary target formations. 
Specifically, the stationary setting is what makes these designs viable: 
in \cite{ZhaobangweiZhaojianingYUxiao_IFAC_2023}, the velocity error used in the adaptive scheme reduces to the agent's own velocity, owing to zero target velocity; 
in \cite{Li-Song-Xie_TCNS_2025}, it is replaced by the deviation from a static
bearing-driven reference that vanishes at the target formation, rendering the bearing kinematics a self-contained subsystem; 
and in \cite{Li-Song_IJRNC_2025}, it is dispensed with altogether by a passivity argument. 
However, none of these methods remains applicable when the leaders are in motion, even at a constant velocity.
A recent work \cite{Cheng_arXiv_2025} considers the formation tracking for EL MASs with leaders of time-varying velocities; nevertheless, this scheme is bearing-based instead of bearing-only, as it requires inter-agent communication and relative-state measurements beyond bearings. 
To the best of our knowledge, bearing-only formation tracking of uncertain EL systems remains an open problem, even for the seemingly simple constant-velocity case.

The challenge in this open problem stems from the simultaneous presence of bearing-only sensing and communication-free condition, giving rise to the following specific difficulties:

\begin{itemize}
    \item (C1) The leaders' states, i.e., positions, velocities, and accelerations, are inaccessible to the followers and cannot be estimated via distributed observers as in \cite{Hong_Auto_2008_Observer_Intro,Cheng_arXiv_2025}, since such schemes inherently rely on relative-position sensing or inter-agent communication.
    \item (C2) As a direct consequence of (C1), the velocity error, a signal indispensable to conventional formation tracking designs, is unavailable for feedback. This precludes both the feedforward compensation and the damping of the unknown time-varying leader velocity, rendering the formation tracking problem challenging. 
    \item (C3) Standard adaptive schemes for EL uncertainties are typically driven by velocity errors. However, obtaining such signals generally requires observer-generated reference velocities, which are unavailable in our setup.
\end{itemize}

To overcome these difficulties, this paper develops a novel bearing-only adaptive control law. To the best of our knowledge, this is the first attempt to achieve formation tracking for uncertain EL MASs with leaders of time-varying velocities under the communication-free and bearing-only constraints. Table~\ref{tab:comparison} compares the proposed approach with closely related existing results, and the contributions of this work are summarized as follows.

1) Distinct from existing bearing-only control schemes for EL systems \cite{ZhaobangweiZhaojianingYUxiao_IFAC_2023,Li-Song-Xie_TCNS_2025,Li-Song_IJRNC_2025}, 
which are confined to stationary formations, this work investigates formations with time-varying velocities. 
To address the problems described in (C1) and (C2), the proposed design leverages the key insight that bearing rates inherently encode projected relative-velocity information, which, combined with the bearing rigidity, provides a rigidity-based damping mechanism to compensate for the unavailable velocity error signal, thereby steering the followers toward the unavailable time-varying velocity of leaders.

2) We develop a bearing-driven adaptive control scheme for EL systems with model uncertainties, extending our previous work \cite{Song_LCSS_2024,Song_TIE_2026} to more practical scenarios. 
To address the challenge in (C3), we construct a purely local and bearing-driven auxiliary state, which, 
integrated with the rigidity-based damping mechanism, generates a surrogate error signal. 
This signal replaces the unavailable velocity error to drive the regressor-based update law, thereby achieving adaptive compensation for EL uncertainties.  
Unlike \cite{Cheng_arXiv_2025}, the proposed method eliminates the need for communication-based observers to estimate reference velocities, while guaranteeing the convergence of the auxiliary state.

3) To accommodate the requirement of sufficient bearing rigidity in the aforementioned damping mechanism, we derive a quantitative perturbation bound on the bearing
Laplacian and leverage it to construct a rigidity-preserving set.
This set can be computed offline and is shown to be forward invariant under the proposed law.
Thus, the bearing rigidity is guaranteed to admit a lower bound,  
thereby ensuring the requisite rigidity via a condition on the initial state rather than imposing it as an assumption as in \cite{Song_LCSS_2024,Song_TIE_2026}.

\begin{table}[!t]
\centering
\caption{Comparison with the most relevant bearing-only formation control results}
\label{tab:comparison}
\footnotesize
\setlength{\tabcolsep}{3.5pt}
\renewcommand{\arraystretch}{1.5}
\begin{tabular*}{\columnwidth}{@{\extracolsep{\fill}}cccc@{}}
\toprule
Work\textsuperscript{1} & Model dynamics& Leader's velocity& Rigidity requirement\\\midrule
\cite{ZhaoShiyu_TAC_2019,Trinh_Auto_2021,Zhaojianing_LCSS_2021,Garanayak_TCNS_2025}
 &  Integrator& Constant& $\mathcal{B}^{*}_{ff}\succ0$ (target only)\\
\cite{Song_LCSS_2024,Song_TIE_2026}
 &   Integrator& Time-varying& Assume $\mathcal{B}_{ff}(t)\succ0$\\
\cite{ZhaobangweiZhaojianingYUxiao_IFAC_2023,Li-Song-Xie_TCNS_2025,Li-Song_IJRNC_2025}
 & EL system& Static leader& $\mathcal{B}^{*}_{ff}\succ0$ (target only)\\
\textbf{This work} & EL system& Time-varying& Guarantee $\mathcal{B}_{ff}(t)\succ0$\\ 
\bottomrule
\end{tabular*}

\vspace{1mm}
\parbox{\columnwidth}{\footnotesize\raggedright
\textsuperscript{1} We consider works under bearing-only and communication-free constraints in this comparison.}
\vspace{-5mm}
\end{table}

The remainder of this paper is organized as follows. Section \ref{Section_2} presents preliminaries and formulates the bearing-only formation tracking problem. Section \ref{Section_3} provides the control law and the stability analysis of the closed-loop system. Section \ref{Section_4} provides numerical simulations to validate the theoretical results, and Section \ref{Section_5} concludes this work. 

\textit{Notation:} 
In this paper, 
for a vector $x\in\mathbb R^{n}$, $\|x\|$ (or $\|x\|_2$ for brevity) and $\|x\|_1$ denote its Euclidean and $1$-norm, respectively. 
For a matrix $A$, $\|A\|$ denotes its spectral norm, and $\lambda_{\min}(A)$, $\lambda_{\max}(A)$ its smallest and largest eigenvalues; $A\succ 0$ means that $A$ is positive definite. 
The symbol $\operatorname{col}(\cdot)$ stacks its arguments into a column vector. 
A function $f:[0,\infty)\to\mathbb R^{n}$ belongs to $\mathcal L_1$ if $\int_0^\infty\|f(t)\|\,dt<\infty$. 
$\operatorname{SGN}(\cdot)$ denotes the set-valued map defined in Section~\ref{Section_2_nonsmooth}.

\section{Preliminaries and Problem Statement}\label{Section_2}
\subsection{Model and Dynamics}
Consider an MAS consisting of \textit{n}agents, where the first $n_{l}$ agents are leaders moving with time-varying velocities, and the remaining $n_{f} =n-n_{l}$ agents are followers governed by Euler--Lagrange dynamics. 
The leader and follower dynamics are described by
\begin{equation} \label{eq:dym} 
\begin{array}{@{}l@{\,\,\,\,}l@{}}
\dot{p}_{i}=v_{r},\quad \dot{v}_{i}=\dot{v}_{r},
&
i=1,\ldots,n_{l}, \\[1pt]
M_{i}(p_{i})\ddot{p}_{i}
+C_{i}(p_i,\dot{p}_{i})\dot{p}_{i}
+G_{i}=u_{i},
&
i=n_{l}+1,\ldots,n ,
\end{array}
\end{equation}
where $p_{i} $, $\dot{p}_{i} $, and $\ddot{p}_{i} \in {\mathbb R}^{d} $ denote the position, velocity, and acceleration of agent $i$, respectively, and $v_{r} $ denotes the leaders' time-varying velocity. $M_{i} (p_{i} )\in {\mathbb R}^{d\times d} $ is the positive definite inertia matrix, $C_{i} (p_i,\dot{p}_{i} )\in {\mathbb R}^{d\times d} $ denotes the unknown Coriolis and centrifugal term, $G_{i} \in {\mathbb R}^{d} $ is the gravitational term, and $u_{i} \in {\mathbb R}^{d} $ is the control input. 

This paper focuses on the design of distributed control laws for the followers. 
The leaders are assumed to be controlled independently so as to generate the reference motions. 
This is standard in the leader--follower bearing-only formation control, where the leaders provide the reference motion while the followers are required to achieve formation tracking using local bearing measurements.

We present the following properties that are well-known and important for the control law design. 

\begin{myproperty}\label{property:skew_symmetric}
The matrix $(\dot{M}_{i} -2C_{i} )$ is skew symmetric, that is, for any $x\in {\mathbb R}^{d} $, one has $x^{T} (\dot{M}_{i} -2C_{i} )x=0$. 
\end{myproperty}
\begin{myproperty}\label{property:linearly_parameterize}
The unknown dynamics in \eqref{eq:dym} can be linearly parameterized, i.e., for any $x$ and $\dot{x}\in {\mathbb R}^{d} $, one has
\begin{equation} \label{eq:linear_para} 
M_{i}(p_i) \dot{x}+C_{i}(p_i,\dot{p}_i) x+G_{i} =Y_{i}({p}_{i},\dot{p}_{i} ,x,\dot{x})\Theta_{i}  ,
\end{equation} 
where $Y_{i}({p}_{i},\dot{p}_{i} ,x,\dot{x})$ is a regressor matrix and $\Theta_{i} $ is an unknown constant parameter vector. 
\end{myproperty}
\begin{myproperty}\label{property:boundedness_of_M}
$M_i$ is symmetric positive definite and satisfies $m_1 I_d \le M_i \le m_2 I_d$ for constants $m_2 \ge m_1 > 0$.
\end{myproperty}
\begin{myproperty}\label{property:boundedness_of_C}
The Coriolis matrix satisfies $\|C_i(p, \dot{p})\| \le c_c \|\dot{p}\|$ for a constant $c_c > 0$.
\end{myproperty}

\subsection{Graph Theory and Formation Configuration}

In this paper, inter-agent communication is not required, and formation tracking is to be achieved solely via bearing measurements. Accordingly, the sensing topology is described as an undirected graph ${\mathit{{\mathcal G}}}=({\mathit{{\mathcal V}}},{\mathit{{\mathcal E}}})$, where ${\mathit{{\mathcal V}\; }}=\{ 1,\ldots ,n\} $ is the vertex set and ${\mathit{{\mathcal E}}}\subseteq {\mathit{{\mathcal V}\; }}\times {\mathit{{\mathcal V}}}$ is the edge set. Specifically, for vertices \textit{i} and \textit{j}, $(i,j)\in {\mathit{{\mathcal E}}}$ means that agent \textit{i }can obtain the bearing information of agent \textit{j}, with agent \textit{j }being a neighbor of agent \textit{i}. Then, the neighbor set of the agent \textit{i }can be defined as ${\mathit{{\mathcal N}}}_{i} =\{ j\in {\mathit{{\mathcal V}}}{\,}|{\,}(i,j)\in {\mathit{{\mathcal E}}}\} $. In addition, since this paper considers the undirected graph, it holds that $(i,j)\in {{\mathcal E}}\Leftrightarrow (j,i)\in {{\mathcal E}}$.

Let $e_{ij}$ and $g_{ij}$ denote the relative position vector and the bearing vector between agent $i$ and agent $j$, respectively, which are defined as
\begin{equation} \label{eq:e_ij_g_ij} 
e_{ij} =p_{j} -p_{i} ,{\rm \; \; }
g_{ij} =\frac{e_{ij}}{\left\| e_{ij} \right\|}.
\end{equation} 

Let $m$ be the number of edges in $\mathcal G$, each assigned an arbitrary orientation and indexed by $k\in\{1,\ldots,m\}$. For the $k$-th edge, directed from agent $i$ to agent $j$, we denote its relative position and bearing vectors by $e_k:=e_{ij}$ and $g_k:=g_{ij}$, respectively. Stacking them yields the aggregated displacement vector $e = \mathrm{col}(e_1, \dots, e_m) \in \mathbb{R}^{md}$ and the aggregated bearing vector $g = \mathrm{col}(g_1, \dots, g_m) \in \mathbb{R}^{md}$. 
To characterize the network topology, the incidence matrix $H = [h_{ki}] \in \mathbb{R}^{m \times n}$ associated with the graph is defined such that $h_{ki} = 1$ if vertex $i$ is the head of the $k$-th edge, $h_{ki} = -1$ if vertex $i$ is its tail, and $h_{ki} = 0$ otherwise.

The configuration is defined as $p={\rm col}(p_{1} ,\ldots, p_{n} )\in {\mathbb R}^{nd} $, which can be partitioned as $p={\rm col}(p_{l} ,p_{f} )$, where $p_{l} $ and $p_{f} $ represent the leader and follower configurations, respectively. Similarly, the velocity vector is defined as $v=\dot{p}={\rm col}(v_{l} ,v_{f} )$. The incidence matrix relates the formation configuration to the displacement vector via 
\begin{equation} \label{eq:def_e} 
e=(H\otimes I_{d} )p=\bar{H}p,{\rm \;\;\; }\dot{e}=(H\otimes I_{d} )v=\bar{H}v .
\end{equation} 

The target formation configuration, denoted by $p^{*} ={\rm col}(p_{l} ,p_{f}^{*} )$, is determined by two constraints: (i) the desired inter-agent bearing $g_{ij}^{*} $ for $(i,j)\in {\mathit{{\mathcal E}}}$, and (ii) the leaders' time-varying positions $p_{l} $. 

For a bearing vector $g_{ij} $, the associated orthogonal projection matrix is defined as $P_{g_{ij} } =I_{d} -g_{ij} g_{ij}^{T} $, which satisfies ${\rm Null}(P_{g_{ij} } )={\rm span}\{ g_{ij} \} $. The bearing Laplacian matrix is defined as ${\mathit{{\mathcal B}}}=[{\mathit{{\mathcal B}}}_{ij} ]\in {\mathbb R}^{dn\times dn} $, where $[{\mathit{{\mathcal B}}}_{ij} ]=0_{d\times d} $ for $i\ne j$, $(i,j)\notin {\mathit{{\mathcal E}}}$; $[{\mathit{{\mathcal B}}}_{ij} ]=-P_{g_{ij} } $ for $i\ne j$, $(i,j)\in {\mathit{{\mathcal E}}}$; and $[{\mathit{{\mathcal B}}}_{ij} ]=\sum _{k\in {\mathit{{\mathcal N}}}_{i} }P_{g_{ik} }  $ for $i=j$, $i\in {\mathit{{\mathcal V}}}$. Partitioning the matrix $\mathcal{B}$ according to the leader and follower sets yields \(\mathcal B
=\bigl(\begin{smallmatrix}
\mathcal B_{ll} & \mathcal B_{lf} \\
\mathcal B_{fl} & \mathcal B_{ff}
\end{smallmatrix}\bigr)\). 

The uniqueness of the formation is characterized by the bearing Laplacian, as described in the following lemma. 

\begin{lem}[{Unique formation condition \cite{ZhaoShiyu_Auto_2016}}]
The formation configuration $p(t)$ is uniquely determined by the bearings $\{g_{ij}\}_{(i,j)\in {{\mathcal E}}}$ and leader positions $p_{l}$ if and only if  $\mathcal B_{ff}$ is positive definite.
\end{lem}

\subsection{Nonsmooth Analysis}\label{Section_2_nonsmooth}

Since the control law proposed in this paper contains discontinuous terms, 
the closed-loop solutions are understood in the Filippov sense. 
Consider $\dot x=f(t,x)$, 
where $f:[0,\infty)\times\mathcal D\to\mathbb R^{N}$ is Lebesgue measurable and locally
essentially bounded on an open set $\mathcal D\subseteq\mathbb R^{N}$. 
The Filippov set-valued map of $f$ is defined as 
\begin{equation}\label{eq:filippov_map}
K[f](t,x)=\bigcap\nolimits_{\delta>0}\bigcap\nolimits_{\mu(S)=0}
\overline{\operatorname{co}}\,f\bigl(t,\mathcal O (x,\delta)\setminus S\bigr),
\end{equation}
where $\mathcal O(x,\delta)$ is the open ball of radius $\delta$ centered at $x$,
$\overline{\operatorname{co}}(\cdot)$ is the convex closure, and the second
intersection is over all sets $S$ of Lebesgue measure zero. 
A Filippov solution on $[0,t_{1})$ is an absolutely continuous $x(\cdot)$ satisfying $\dot x\in K[f](t,x)$ for almost all $t$. 
Such a solution exists for every initial condition in $\mathcal D$. 
Moreover, boundedness of a solution within $\mathcal D$ guarantees its existence on $[0,\infty)$. 

For componentwise sign functions, 
the Filippov set-valued map is $K[\operatorname{sgn}](y)=\operatorname{SGN}(y)$, 
where $\operatorname{SGN}_{i}(y)=\{\operatorname{sgn}(y_{i})\}$ for $y_{i}\neq0$ and $\operatorname{SGN}_{i}(y)=[-1,1]$ for $y_{i}=0$. 
Since the Lyapunov function designed in this paper is
continuously differentiable, its Clarke generalized gradient reduces to the
singleton $\{\nabla V\}$, and the chain rule for Filippov systems is specialized in the following lemma.

\begin{lem}[Nonsmooth calculus \cite{Shevitz_TAC_1994}]\label{lem:nonsmooth}
Let $x(\cdot)$ be a Filippov solution of $\dot x=f(t,x)$. Then the
following statements hold:

(i) If $V(t,x)$ is continuously differentiable, then $t\mapsto V(t,x(t))$ is
absolutely continuous, and its derivative exists and satisfies, for almost all $t$,
\begin{equation}\label{eq:setvalued_derivative}
\dot V\in\dot{\tilde V}(t,x)
:=\bigl\{\partial_{t}V+\nabla_{x}V^{T}w \bigm| w\in K[f](t,x)\bigr\}.
\end{equation}
Moreover,  $\max\dot{\tilde V}(t,x)\le0$ implies that $V$ is nonincreasing.

(ii) For any $y\in\mathbb R^{N}$, the set-valued product $y^{T}\operatorname{SGN}(y)$
reduces to the singleton $\{\|y\|_{1}\}$; that is, $y^{T}\sigma=\|y\|_{1}$ holds for all $\sigma\in\operatorname{SGN}(y)$.
\end{lem}

\subsection{Problem Statement}
We consider a communication-free MAS in which the $n_l$ leaders move with \emph{time-varying} velocities $v_r(t)$, while the $n_f$ followers are governed by the uncertain EL dynamics \eqref{eq:dym}. The setting under study is distinguished by the following two features. 

\textit{(F1) No inter-agent communication is available}: no follower has access to the leaders' position, velocity, or acceleration via communication, nor can it estimate them through a distributed observer, since such schemes inherently rely on state exchange among neighbors. 

\textit{(F2) The followers' perception is local and purely directional}: follower $i$ can only measure the inter-agent bearings $\{g_{ij}(t)\}_{j\in\mathcal N_i}$ and their rates $\{\dot g_{ij}(t)\}_{j\in\mathcal N_i}$, which retain only the directions of the
relative positions while discarding the inter-agent distances $\|p_{i}-p_{j}\|$. 

In this paper, we adopt the following two assumptions, concerning the target formation and the leaders' motion, respectively. 

\begin{assum}\label{assum:positive-definite}
The target formation is uniquely determined by the desired bearing vectors and leaders' positions, i.e., ${{\mathcal B}}_{ff}^\ast \succ 0$. 
\end{assum}

\begin{assum}\label{assum:acc-bound}
The leader velocity $v_r(t)$ is bounded, and its acceleration is bounded by a known constant $\delta_a$; that is, $\sup_{t\ge0}\|v_r(t)\|_2<\infty$ and $\|\dot v_r(t)\|_2\le\delta_a$ for all $t\ge0$.
\end{assum}

\textit{Problem}: Consider an MAS whose leaders move with time-varying velocity, under features (F1)--(F2) and 
Assumptions~\ref{assum:positive-definite}--\ref{assum:acc-bound}. This paper aims to design a bearing-only control law $u_i$ for each follower $i$ that relies exclusively on the locally available measurements $\{g_{ij}(t),\dot g_{ij}(t)\}_{j\in\mathcal N_i}$ and the follower's own states $p_i$ and $v_i$, such that, for all initial conditions within an explicitly characterized set,  

(1) the follower velocity tracking error converges to zero, i.e., $\lim_{t\to\infty}(v_i(t)-v_r(t))=0$; and

(2) the follower position error is uniformly ultimately bounded, i.e., $\limsup_{t\to\infty}\|p(t)-p^{*}(t)\|\le \bar{b}$, where the ultimate bound $\bar b$ can be rendered arbitrarily small by increasing the control gains. 

In this case, the closed-loop system is said to achieve local practical formation tracking under bearing-only and communication-free conditions. 

\begin{rem}
Under the communication-free constraint, the leaders' velocity $v_r$ is completely inaccessible to the followers. Traditional estimation methods, such as distributed observers, are rendered invalid since they require state information exchange among neighbors. Consequently, the velocity tracking error $\tilde v_i=v_i-v_r$ cannot be obtained for feedback design. 
The lack of the velocity tracking error presents dual challenges: (i) it precludes velocity feedforward compensation, and (ii) it prevents the design of regressor-based adaptive schemes for uncertain EL dynamics, which inherently rely on this velocity error to drive parameter update laws.
Note that the former obstacle renders time-varying velocity tracking highly challenging, even for single- or double-integrators without uncertainties. Meanwhile, the latter obstacle explains why even the constant-velocity case has remained an open problem for bearing-only EL systems. The confluence of these two challenges constitutes the technical difficulty of this work. 
\end{rem}

\section{Main Results}\label{Section_3}
\subsection{Formation Tracking Control for EL Followers}
The control law for follower $i$ is designed as 
\begin{subequations} \label{eq:local_control_law}
\begin{align}
u_i ={}&
 k_p \sum_{j\in \mathcal N_i}
 \bigl(g_{ij}(t)-g_{ij}^{*}\bigr)
 + k_v \sum_{j\in \mathcal N_i}\dot g_{ij}(t)+ Y_i\hat{\Theta}_i
 \notag\\
&
 + \gamma \sum_{j\in \mathcal N_i}
 P_{g_{ij}(t)}\operatorname{sgn}\bigl(\dot g_{ij}(t)\bigr)
 - \beta \operatorname{sgn}(v_i-\eta_i),
 \label{eq:local_controller_u}
\\
\dot{\hat{\Theta}}_i ={}&
 -\Lambda_i Y_i^{T}(p_i,\dot p_i,\eta_i,\dot\eta_i)(v_i-\eta_i),
 \label{eq:local_controller_theta}
\\
\dot{\eta}_i ={}&
 k_p \sum_{j\in \mathcal N_i}
 \bigl(g_{ij}(t)-g_{ij}^{*}\bigr)
 + k_v \sum_{j\in \mathcal N_i}\dot g_{ij}(t)
 \notag\\
&\qquad\qquad\qquad\qquad
 + \gamma \sum_{j\in \mathcal N_i}
 P_{g_{ij}(t)}
 \operatorname{sgn}\bigl(\dot g_{ij}(t)\bigr),
 \label{eq:local_controller_eta}
\end{align}
\end{subequations}
where $k_{p} $ and $k_{v} $ are positive control gains, $\Lambda _{i} $ is an adaptation gain matrix with positive diagonal elements, 
and the design parameters $\gamma$ and $\beta$ are chosen according to the conditions in \eqref{eq:xi_condition}. 
The purely local auxiliary variable $\eta_i$, governed by \eqref{eq:local_controller_eta} with an arbitrary initialization $\eta_i(0)$, serves as a bearing-driven surrogate for the unavailable leader velocity.
The regressor $Y_i(p_i,\dot p_i,\eta_i,\dot\eta_i)$ is constructed based on Property~\ref{property:linearly_parameterize}, and $\hat\Theta_i$ denotes the estimate of the unknown parameter vector $\Theta_i$. 
Consequently, the control law \eqref{eq:local_control_law} is implementable using only
$\{g_{ij},\dot g_{ij}\}_{j\in\mathcal N_i}$ and the follower's own state $(p_i,v_i)$, without any inter-agent communication or leaders' state information.

Let $u_f = \mathrm{col}(u_{n_l+1}, \dots, u_n)$, $v_f = \mathrm{col}(v_{n_l+1}, \dots, v_n)$, and $\eta_f = \mathrm{col}(\eta_{n_l+1}, \dots, \eta_n)$ denote the stacked vectors of the followers, where the followers are indexed from $n_l+1$ to $n$ since the first $n_l$ agents are the leaders. 
Recalling the aggregated bearing vector $g$ defined in Section~\ref{Section_2} and letting $g^*$ be its desired vector, the control law \eqref{eq:local_control_law} can be expressed in the following compact matrix form:
\begin{subequations}\label{eq:compact_control_law}
\begin{align}
u_f={}&
-k_p\bar R^T(g-g^*)
-k_v\bar R^T\dot g
\notag\\
&
-\gamma\bar R^T{\rm diag}(P_{g_k}){\rm sgn}(\dot g)
-\beta{\rm sgn}(v_f-\eta_f)
+Y_f\hat\Theta_f,
\label{eq:compact_control_law_u}
\\
\dot{\hat\Theta}_f={}&
-\Lambda_fY_f^T(v_f-\eta_f),
\label{eq:compact_control_law_theta}
\\
\dot\eta_f={}&
-k_p\bar R^T(g-g^*)
-k_v\bar R^T\dot g
-\gamma\bar R^T{\rm diag}(P_{g_k}){\rm sgn}(\dot g),
\label{eq:compact_control_law_eta}
\end{align}
\end{subequations}
where $\bar{R} = R \otimes I_d $, with $R \in \mathbb{R}^{m \times n_f}$ being the follower-related incidence matrix formed by extracting the last $n_f$ columns of $H$. Moreover, 
$\hat\Theta_f={\rm col}(\hat\Theta_{n_l+1},\ldots,\hat\Theta_n)$, 
$Y_f={\rm diag}\{Y_{n_l+1},\ldots,Y_n\}$, 
$\Lambda_f={\rm diag}\{\Lambda_{n_l+1},\ldots,\Lambda_n\} $, and ${\rm diag}(P_{g_{k} } )={\rm diag}\{ P_{g_{1} } ,\ldots ,P_{g_{m} } \} $. 

\subsection{Stability Analysis}
Before presenting the main result, we introduce the following important lemmas. 

\begin{lem}[Bearing error lower bound \cite{ZhaoShiyu_TAC_2019}]\label{lem:ZHAO_positive_lemma}
For a bearing-only formation with no coinciding agents, it follows that 
\begin{equation} \label{eq:bearing-error-bound} 
\begin{aligned}
e^{T}(g-g^{*})
\ge
\frac
{\tilde p_{f}^{T}\mathcal{B}_{ff}^{*}\tilde p_{f}}
{2\max\nolimits_k\|e_k\|} \ge
\frac
{\lambda_{\min}(\mathcal{B}_{ff}^{*})}
{2\max\nolimits_k\|e_k\|}
\|\tilde{p}_f\|^{2}
\ge 0,
\end{aligned}
\end{equation}
where $\tilde{p}=p-p^\ast$ and $\tilde{p}_f=p_f-p_f^\ast$ denote the position error vector and its component for the followers, respectively.
\end{lem}

Recall that  $R\in\mathbb R^{m\times n_f}$ is the follower-related incidence matrix defined above, and let $r_k^T$ denote the $k$-th row of $R$. For each edge $k\in\{1,\ldots,m\}$, define $ N_{k}=(r_k^T\otimes I_d)(\mathcal B_{ff}^\ast)^{-1/2}$. Based on these formulations, we provide the following lemma.

\begin{lem}\label{lem:local_preservation}
Let $\mu>0$ satisfy $\mu\|(g_k^\ast)^T N_{k}\|<\|e_k^\ast\|$ for $k=1,\ldots,m$, and define
\begin{equation}\label{eq:def_epsilon}
    \varepsilon_k(\mu)=\frac{\mu\|P_{g_k^\ast}N_{k}\|}{\|e_k^\ast\|-\mu\|(g_k^\ast)^T N_{k}\|}.
\end{equation}

Let $L_\varepsilon(\mu)$ be the associated bearing-perturbation matrix
\begin{equation}\label{eq:def_matrix_L}
L_\varepsilon(\mu)=
\sum\nolimits_{k=1}^{m}
\varepsilon_k(\mu)\,r_k r_k^{T}.
\end{equation}

Then, for every{\,} $\tilde p_{f}$ in the open ellipsoidal sublevel set
\begin{equation}\label{eq:Omega_def}
\Omega_c:=
\left\{\tilde p_{f}\in\mathbb R^{dn_f}\ \middle|\ 
\tilde p_{f}^T\mathcal B_{ff}^\ast \tilde p_{f}<\mu^{2}
\right\},
\end{equation}
the actual bearing Laplacian of the concerned system satisfies
\begin{equation} \label{eq:Bff_inequality} 
\lambda_{\min}(\mathcal B_{ff})\ge \underline\lambda(\mu):=
\lambda_{\min}(\mathcal B_{ff}^\ast)
-\lambda_{\max}(L_\varepsilon(\mu)),
\end{equation}
and there always exists $\bar\mu > 0$ such that $\underline{\lambda}(\mu) > 0$ for $\mu \in (0,\bar\mu)$.
\end{lem}

\textit{Proof:}
The proof of Lemma~\ref{lem:local_preservation} is presented in Appendix \ref{Appendix:proof_local}. 

\begin{thm}\label{thm:main}
Consider a communication-free multi-agent system \eqref{eq:dym} under Assumptions \ref{assum:positive-definite}--\ref{assum:acc-bound}, in which the leaders move with time-varying velocities, the followers are described by the EL model with uncertainties, and no follower has access to the leaders' position or velocity. Let the followers be driven by control law \eqref{eq:local_control_law}, using only the inter-agent bearing measurements $\{ g_{ij} ,\dot{g}_{ij} \} _{j\in {\mathit{{\mathcal N}}}_{i} } $ and their own states $\{p_{i} ,v_i\}$. 

Let $\mu \in (0,\bar\mu)$ be a fixed constant such that $\underline{\lambda}(\mu) > 0$. 
For the Lyapunov function $V(t)$ defined in \eqref{eq:lyapunov}, 
suppose that the initial condition satisfies
\begin{equation}\label{eq:initial-condition}
V(0)<\bar{V}(\mu)
:=
\frac{k_p \mu^{2}}
{
2\max_{k}
\left(\|e_k^\ast\|+\mu\|N_{k}\|\right)
},
\end{equation}
and the design parameters $\gamma$ and $\beta$ satisfy
\begin{equation}\label{eq:xi_condition}
\gamma >
\frac{\sqrt {n_f}\, \delta_a}{\sqrt{\underline{\lambda}(\mu)}} , {\,\,}
\beta   > \sqrt{n_{f} } \delta _{a}.
\end{equation}

Then, the follower velocity tracking error converges to zero, and the position tracking error is uniformly ultimately bounded (UUB), where the ultimate bound is explicitly quantified by \eqref{pe_boundedness} and can be rendered arbitrarily small by sufficiently increasing the control gain $k_p$.
\end{thm}

\textit{Proof:}
The Lyapunov function candidate is chosen as follows: 
\begin{equation} \label{eq:lyapunov}
\begin{aligned}
V ={}&
k_p e^{T}(g-g^{*})
+\frac{1}{2}(v_f-\eta_f)^{T}M_{f}(v_f-\eta_f) \\
&\quad
+\frac{1}{2}
(\eta_f-\bar{v}_{r})^{T}
(\eta_f-\bar{v}_{r}) 
+\frac{1}{2}
\tilde{\Theta}_{f}^{T}
\Lambda_{f}^{-1}
\tilde{\Theta}_{f},
\end{aligned}
\end{equation}
where $\bar{v}_{r} (t)=\mathbf{1}_{n_f} \otimes v_{r} (t)$, $M_{f}={\rm diag}\{ M_{n_l+1} ,\ldots ,M_{n} \} $, and $\tilde{\Theta }_{f}=\Theta_{f} -\hat{\Theta }_{f}$. By Lemma~\ref{lem:ZHAO_positive_lemma}, one has $e^{T} (g-g^{*} )\ge 0$ and $V\ge 0$.

Since the control law \eqref{eq:local_control_law} is discontinuous, the closed-loop system features a discontinuous right-hand side. Let $x$ denote the closed-loop system state and $f(t,x)$ the corresponding vector field determined by \eqref{eq:dym} and \eqref{eq:compact_control_law}. Since $V$ in \eqref{eq:lyapunov} is continuously differentiable, Lemma~\ref{lem:nonsmooth} implies that $V$ is absolutely continuous along any Filippov solution and satisfies, for almost all $t$,

\begin{equation}\label{eq:set_valued_derivative_V}
\dot V\in\dot{\tilde V}
=\bigl\{\partial_{t}V+\nabla_{x}V^{T}w \bigm| w\in K[f](t,x)\bigr\}.
\end{equation}

It then follows that
\begin{equation} \label{eq:dot_V}
\begin{aligned}
\dot{\tilde{V}}
={}&
\underbrace{
k_{p}e^{T}\dot{g}
+k_{p}(g-g^{*})^{T}\bar{H}v
}_{\dot{\tilde{V}}_{1}}
+
\underbrace{
\xi_{f}^{T}M_{f}\dot{\xi}_{f}
+\tfrac{1}{2}\xi_{f}^{T}\dot{M}_{f}\xi_{f}
}_{\dot{\tilde{V}}_{2}} \\
&\quad
\underbrace{
+(\eta_{f}-\bar{v}_{r})^{T}
(\dot{\eta}_{f}-\dot{\bar{v}}_{r})
}_{\dot{\tilde{V}}_{3}}
\underbrace{
-\tilde{\Theta}_{f}^{T}\Lambda_{f}^{-1}
\dot{\hat{\Theta}}_{f}
}_{\dot{\tilde{V}}_{4}} {\,\,},
\end{aligned}
\end{equation}
where $\xi_{f}=v_f-\eta_f$, and the first term in $\dot{\tilde{V}}_{1} $ is derived as
\begin{equation} \label{eq:dot_V1}
\begin{aligned}
e^{T}\dot g
={}& \sum_{k=1}^{m}
\frac{
\dot e_k^{T}e_k
-
e_k^{T}e_k
\left(e_k^{T}e_k\right)^{-1}
\dot e_k^{T}e_k
}{
\|e_k\|_2
} = 0.
\end{aligned}
\end{equation}

By substituting the EL dynamics \eqref{eq:dym} and the control law \eqref{eq:compact_control_law_u}, and letting $C_{f}=\mathrm{diag}\{C_{n_l+1},\dots,C_n\} $ and $G_f=\text{col}(G_{n_l+1},\dots,G_n)$, we can expand the term $\dot{\tilde{V}}_{2}$ as 
\begin{equation} \label{eq:V2dot}
\begin{aligned}
\dot{\tilde{V}}_{2}
={}& \xi_f^{T}(u_f-C_f v_f-G_f)
    -\xi_f^{T}M_f\dot{\eta}_f
    +\tfrac{1}{2}\xi_f^{T}\dot{M}_f\xi_f \\
={}& \xi_f^{T}u_f
    -\xi_f^{T}(C_f\xi_f+C_f\eta_f+G_f)\\
    &\qquad \qquad  -\xi_f^{T}M_f\dot{\eta}_f
    +\tfrac{1}{2}\xi_f^{T}\dot{M}_f\xi_f \\
={}& \tfrac{1}{2}\xi_f^{T}(\dot{M}_f-2C_f)\xi_f
    -\xi_f^{T}(M_f\dot{\eta}_f+C_f\eta_f+G_f) \\
&
    -k_p\xi_f^{T}\bar R^{T}(g-g^{*})
    -k_v\xi_f^{T}\bar R^{T}\dot{g}
    +\xi_f^{T}Y_f\hat{\Theta}_f \\
&
    -\gamma \xi_f^{T}\bar R^{T}
    \operatorname{diag}(P_{g_k})
    \operatorname{SGN}(\dot{g})
    -\beta \xi_f^{T}\operatorname{SGN}(\xi_f) \\
={}&
    -k_p\xi_f^{T}\bar R^{T}(g-g^{*})
    -k_v\xi_f^{T}\bar R^{T}\dot{g}
    -\xi_f^{T}Y_f\tilde{\Theta}_f \\
&
    -\gamma \xi_f^{T}\bar R^{T}
    \operatorname{diag}(P_{g_k})
    \operatorname{SGN}(\dot{g})
    -\beta \xi_f^{T}\operatorname{SGN}(\xi_f) ,
\end{aligned}
\end{equation}
where the first equality invokes the EL dynamics
\eqref{eq:dym} together with $M_f\dot{\xi}_f=M_f\dot v_f-M_f\dot{\eta}_f$;
the second equality uses $v_f=\eta_f+\xi_f$; the third equality substitutes
the control law \eqref{eq:compact_control_law_u}; and the last equality follows from $\xi_f^{T}(\dot M_f-2C_f)\xi_f=0$ stated in Property~\ref{property:skew_symmetric} and $M_f\dot{\eta}_f+C_f\eta_f+G_f=Y_f\Theta_f$ stated in
Property~\ref{property:linearly_parameterize}.

Note that $\operatorname{SGN}(\cdot)$ denotes the Filippov set-valued map employed to treat the discontinuous components in the control law (see Section \ref{Section_2_nonsmooth} for a detailed description). 
\begin{rem}
It should be noted that the last equality in \eqref{eq:V2dot}, as well as the subsequent derivations, involve algebraic operations of terms defined through the set-valued map $\operatorname{SGN}(\cdot)$. 
Such operations are legitimate since the discontinuous component $\gamma\bar R^{T}{\rm diag}(P_{g_{k}}){\rm sgn}(\dot g)$ and the signal $\dot\eta_f$ are computed and shared throughout the control input \eqref{eq:compact_control_law_u}, the adaptation law \eqref{eq:compact_control_law_theta} and the regressor matrix.
Consequently, despite involving set-valued map $\operatorname{SGN}(\cdot)$, $\dot\eta_f$ refers to the same quantity wherever it appears, and the equalities involved hold pointwise for almost all $t$ along the closed-loop solutions. 
\end{rem}

Under the auxiliary dynamics \eqref{eq:compact_control_law_eta}, $\dot{\tilde{V}}_{3}$ is derived as
\begin{equation} \label{eq:V3dot} 
\begin{aligned}
\dot{\tilde{V}}_{3}
={}& (\eta_f -\bar{v}_{r} )^{T}
      (\dot{\eta }_f-\dot{\bar{v}}_{r} ) \\
={}& -k_{p} \tilde{\eta }_{f}^{T}
      \bar R^{T}
      \left(g-g^{*} \right)
      -k_{v} \tilde{\eta }_{f}^{T}
      \bar R^{T} \dot{g} \\
&
      -\gamma \tilde{\eta }_{f}^{T}
      \bar R^{T}
      {\rm diag}(P_{g_{k} } )
      \operatorname{SGN}(\dot{g})
      -\tilde{\eta }_{f}^{T}
      \dot{\bar{v}}_{r} ,
\end{aligned}
\end{equation} 
where $\tilde{\eta }_{f}=\eta_{f} -\bar{v}_{r} $ is introduced for notational brevity. 

Furthermore, recalling the relation $\dot{g}_k=\big(P_{g_k}/\|e_k\|\big)\dot{e}_k$, we present the following identity, which is essential for the subsequent analysis:
\begin{equation} \label{eq:sign_transfer} 
\begin{aligned}
\operatorname{SGN}(\dot{g})
={}& \operatorname{SGN}({\rm diag}(P_{g_{k} } )
      {\rm diag}(\left\| e_{k} \right\| )^{-1}
      \dot{e}) \\
={}& \operatorname{SGN}({\rm diag}(P_{g_{k} } )
      \dot{e}) \\
={}& \operatorname{SGN}({\rm diag}(P_{g_{k} } )
      \bar{H}v) \\
={}& \operatorname{SGN}({\rm diag}(P_{g_{k} } )
      \bar{R}\tilde{v}_f),
\end{aligned}
\end{equation}
where ${\rm diag}(\|e_{k}\|)=\operatorname{diag}\{\|e_{1}\|,\ldots,\|e_{m}\|\}\otimes I_{d}$.
The second equality follows from the invariance of the sign map under the positive diagonal scaling ${\rm diag}(\|e_k\|)^{-1}$; 
and the last equality follows from $\bar H v=\bar H\tilde v=\bar R\tilde v_f$. Here $\tilde v=v-\mathbf 1_{n}\otimes v_r=\operatorname{col}(0,\tilde v_f)$, 
and $\bar H(\mathbf 1_{n}\otimes v_r)=(H\mathbf 1_{n})\otimes v_r=0$, while $\tilde v_f=v_f-\mathbf 1_{n_f}\otimes v_r$.

Substituting the adaptive law \eqref{eq:compact_control_law_theta} into $\dot{\tilde V}_4$ in \eqref{eq:dot_V} leads to
\begin{equation}\label{eq:V4dot}
\begin{aligned}
\dot{\tilde V}_4=
{}&-\tilde\Theta_f^{T}\Lambda_f^{-1}\dot{\hat\Theta}_f\\
={}&-\tilde\Theta_f^{T}\Lambda_f^{-1}\bigl(-\Lambda_f Y_f^{T}\xi_f\bigr)
=\xi_f^{T}Y_f\tilde\Theta_f.
\end{aligned}
\end{equation}

Substituting \eqref{eq:dot_V1}--\eqref{eq:V4dot} into \eqref{eq:dot_V} yields
\begin{equation} \label{eq:Vdot-reduced} 
\begin{aligned}
\dot{\tilde{V}}
={}& k_{p} (g-g^{*} )^{T}\bar{H}v
      -k_{p} \xi_f ^{T}\bar R^{T}
      \left(g-g^{*} \right)
      -k_{v} \xi_f ^{T}\bar R^{T}\dot{g} \\
&
      -\gamma \xi_f ^{T}\bar R^{T}
      {\rm diag}(P_{g_{k} } )
      \operatorname{SGN}(\dot{g})
      -\beta \xi_f ^{T}\operatorname{SGN}(\xi_f ) \\
&
      -k_{p} \tilde{\eta }_{f} ^{T}\bar R^{T}
      \left(g-g^{*} \right)
      -k_{v} \tilde{\eta }_{f} ^{T}\bar R^{T}\dot{g} \\
&
      -\gamma \tilde{\eta }_{f} ^{T}\bar R^{T}
      {\rm diag}(P_{g_{k} } )
      \operatorname{SGN}(\dot{g}) 
      -\tilde{\eta }_{f}^{T}\dot{\bar{v}}_{r} \\
={}& -k_{v} v^{T}\bar{H}^{T}\dot{g}
      -\tilde{v}_f^{T}\dot{\bar{v}}_{r}
      +\xi_f ^{T}\dot{\bar{v}}_{r}  
      -\beta \xi_f ^{T}\operatorname{SGN}(\xi_f )\\
&
      -\gamma \tilde {v}_f^{T}\bar R^{T}
      {\rm diag}(P_{g_{k} } )
      \operatorname{SGN}({\rm diag}(P_{g_{k} } )
      \bar R\tilde{v}_f) ,
\end{aligned}
\end{equation}
where the final equality follows from the identity $\xi_f+\tilde{\eta }_{f}=\tilde v_f$. 
Specifically, this substitution facilitates the combination of the terms in $\dot{\tilde V}_2$ and $\dot{\tilde V}_3$: 
the proportional ($k_p$) terms cancel against the $\dot{\tilde V}_1$ contribution via $\bar R\tilde v_f=\bar H v$, 
the derivative ($k_v$) terms reduce to $-k_v v^{T}\bar H^{T}\dot g$, 
the sign term is reformulated via \eqref{eq:sign_transfer}, 
and the acceleration term is decomposed into $-\tilde v_f^{T}\dot{\bar v}_r +\xi_f^{T}\dot{\bar v}_r$ by $\tilde{\eta }_{f}=\tilde v_f-\xi_f$.

By Lemma~\ref{lem:nonsmooth}, the time derivative of the Lyapunov function \eqref{eq:lyapunov} satisfies
\begin{equation} \label{eq:dot_V_almost_final}
\begin{aligned}
\dot{V}
={}& -k_v v^{T}\bar{H}^{T}\dot{g}
    -\tilde{v}_{f}^{T} (\mathbf 1_{n_{f}} \otimes \dot{v}_{r})
    +\xi_{f}^{T}(\mathbf 1_{n_{f}} \otimes \dot{v}_{r}) \\
&
    -\gamma \bigl\|\operatorname{diag}(P_{g_k})\bar {R}\tilde{v}_{f}\bigr\|_{1}
    -\beta \|\xi_{f}\|_{1} \\
\le{}& -k_v v^{T}\bar{H}^{T}\dot{g}
    -\gamma \bigl\|\operatorname{diag}(P_{g_k})\bar {R}\tilde{v}_{f}\bigr\|_{2} \\
&
    +\|\tilde{v}_{f}\|_{2}\|\mathbf 1_{n_{f}} \otimes \dot{v}_{r}\|_{2}
    +\|\xi_{f}\|_{2}\|\mathbf 1_{n_{f}} \otimes \dot{v}_{r}\|_{2}
    -\beta \|\xi_{f}\|_{2},
\end{aligned}
\end{equation}
where the equality holds due to the symmetry of $\operatorname{diag}(P_{g_k})$, which implies 
$\tilde v_f^{T}\bar R^{T}\operatorname{diag}(P_{g_k})
=(\operatorname{diag}(P_{g_k})\bar R\tilde v_f)^{T}$.  
Then, invoking Filippov's calculus yields the $\ell_1$-norms,
and the inequality follows from $\|\cdot\|_1\ge\|\cdot\|_2$ together with the Cauchy--Schwarz bounds on the two acceleration cross-terms. 

By \cite[Lemma~2]{ZhaoShiyu_Auto_2016}, one has $\mathcal B=\bar{H}^{T}\operatorname{diag}(P_{g_k})\bar{H}$. 
Combining this with the partitioning of $\mathcal B$ yields $\mathcal B_{ff}=\bar{R}^{T}\operatorname{diag}(P_{g_k})\bar{R}$. 
Recalling that $\operatorname{diag}(P_{g_k})$ is symmetric and idempotent, one obtains 
\begin{equation}\label{eq:symmetric_idempotent_matrix}
\begin{aligned}
\|\operatorname{diag}(P_{g_k})\bar {R}\tilde{v}_f\|_2^2
={}& \tilde{v}_f^T \bar {R}^T \operatorname{diag}(P_{g_k})^2 \bar {R}\tilde{v}_f \\
={}& \tilde{v}_f^T \bar {R}^T \operatorname{diag}(P_{g_k}) \bar {R}\tilde{v}_f \\
={}& \tilde{v}_f^T \mathcal B_{ff}\tilde{v}_f ,
\end{aligned}
\end{equation}
which implies that
\begin{equation}\label{eq:norm_sqrt_inequality}
\begin{aligned}
\|\operatorname{diag}(P_{g_k})\bar {R}\tilde v_f\|
={}&\sqrt{\tilde v_f^T \mathcal{B}_{ff}\tilde v_f}\\
\ge{}&
\sqrt{\lambda_{\min}(\mathcal{B}_{ff})}\,\|\tilde v_f\|.
\end{aligned}
\end{equation}

Combining \eqref{eq:norm_sqrt_inequality},  $\bar{H}v=\dot e$, $\dot{g}=\operatorname{diag}(P_{g_k} / \|e_k\|)\dot e$, and  $\left\| \dot{v}_{r} (t)\right\|  \le \delta _{a} $ as stated in Assumption~\ref{assum:acc-bound}, it follows that
\begin{equation}\label{eq:Vdot-bound}
\begin{aligned}
\dot V
\le{}& -k_v v^{T}\bar{H}^{T}\dot{g}
    -\gamma \bigl\|\operatorname{diag}(P_{g_k})\bar {R}\tilde{v}_{f}\bigr\|_{2}\\
&\quad\quad
+\sqrt{n_f}\delta_a\|\tilde v_f\|
+\sqrt{n_f}\delta_a\|\xi_{f}\|
-\beta\|\xi_{f}\| \\
\le{}& -k_v \dot e^T
\operatorname{diag}\!\left(\frac{P_{g_k}}{\|e_{k}\|}\right)\dot e
-\gamma \sqrt{\lambda_{\min}(\mathcal{B}_{ff})}\,\|\tilde v_f\| \\
& \quad\quad +\sqrt{n_f}\delta_a\|\tilde v_f\|
+\sqrt{n_f}\delta_a\|\xi_{f}\|
-\beta\|\xi_{f}\|  \\
\le{}& -k_v \dot e^T
\operatorname{diag}\!\left(\frac{P_{g_k}}{\|e_{k}\|}\right)\dot e -(\beta-\sqrt{n_f}\delta_a)\|\xi_{f}\| \\
& \quad\quad\quad
-\left(
\gamma\sqrt{\lambda_{\min}(\mathcal{B}_{ff})}
-\sqrt{n_f}\delta_a
\right)\|\tilde v_f\|,
\end{aligned}
\end{equation}
where the first term is nonpositive since each orthogonal projection $P_{g_{k}}$ is positive semidefinite. 
We now show that the remaining two terms are nonpositive for all states satisfying $\tilde p_{f}\in\Omega_c$.
By Lemma~\ref{lem:local_preservation}, $\tilde p_{f}\in\Omega_c$ implies $\lambda_{\min}(\mathcal B_{ff})\ge\underline\lambda(\mu)>0$. 
Combining this with the gain conditions \eqref{eq:xi_condition} yields  $\beta-\sqrt{n_f}\delta_a>0$ and $\gamma\sqrt{\lambda_{\min}(\mathcal B_{ff})}
\ge\gamma\sqrt{\underline\lambda(\mu)}>\sqrt{n_f}\delta_a$. 
Consequently, both coefficients in the last two terms of \eqref{eq:Vdot-bound} are positive, leading to
\begin{equation}\label{eq:Vdot-cond}
\dot V\le 0,\quad\text{whenever } \tilde p_{f}\in\Omega_c .
\end{equation}

Note that \eqref{eq:Vdot-cond} is conditional; that is, 
it holds only for the states within $\Omega_c$. 
The following proposition establishes the positive invariance of the set $\Omega_c$ with its proof given in Appendix~\ref{Appendix:proof_invariant}. 

\begin{prop}
\label{prop:Omega_invariant}
Consider the closed-loop system described in Theorem~\ref{thm:main} under the control law \eqref{eq:local_control_law}. If conditions \eqref{eq:initial-condition} and \eqref{eq:xi_condition} are satisfied, then the set {\,}$\Omega_c$ defined in \eqref{eq:Omega_def} is positively invariant.
\end{prop}

Proposition~\ref{prop:Omega_invariant} guarantees that $\Omega_c$ is positively
invariant, i.e., $\tilde p_{f}(t)\in\Omega_c$ for all $t\ge0$. Consequently,
Lemma~\ref{lem:local_preservation} yields
$\lambda_{\min}(\mathcal B_{ff}(t))\ge\underline\lambda(\mu)>0$ for all $t\ge0$.
Let $c_v=\gamma\sqrt{\underline{\lambda}(\mu)}-\sqrt{n_f}\delta_a$ and $c_\xi=\beta-\sqrt{n_f}\delta_a$, both positive by \eqref{eq:xi_condition}.
Then, the bound \eqref{eq:Vdot-bound} implies that
$\dot V\le -c_v\|\tilde v_f\|-c_\xi\|\xi_f\|\le 0$ for almost all $t\ge0$.
Integrating this along the Filippov solutions of the closed-loop system leads to
$c_v\!\int_{0}^{T}\!\|\tilde v_f(t)\|\,dt+c_\xi\!\int_{0}^{T}\!\|\xi_f(t)\|\,dt
\le V(0)-V(T)\le V(0)$, 
and letting $T\to\infty$ yields $\tilde v_f\in \mathcal L_1$ and $\xi_f\in \mathcal L_1$.

Moreover, since $\dot V \le 0 $, the Lyapunov function \eqref{eq:lyapunov} is bounded, implying that $e^{T} (g-g^{*} )$, $\xi_f $, $\tilde{\eta }_{f}$, and $\tilde{\Theta }_{f}$ are bounded. Based on the lower-bound analysis of $e^{T} (g-g^{*} )$ in \cite[Corollary~2]{ZhaoShiyu_TAC_2019}, we further obtain that $\tilde{p}$ is bounded. 

From the positive invariance of set $\Omega_c$ and Lemma~\ref{lem:local_preservation}, we obtain that there exists $\underline e>0$ such that $\|e_k(t)\|\ge \underline e$. Therefore, $\dot g =\operatorname{diag}(P_{g_k} / \|e_k\|)\bar H v
=\operatorname{diag}(P_{g_k} / \|e_k\|) \bar R\tilde v_f $ is bounded since
$\tilde v_f=\xi_f+\tilde\eta_f$ is bounded. Then, from the auxiliary dynamics \eqref{eq:compact_control_law_eta}, $\dot\eta_f$ is bounded, implying that $\dot{\tilde\eta}_{f}=\dot\eta_f-\mathbf 1_{n_f}\otimes \dot v_r$ is bounded.

According to standard measurable selection arguments \cite{Filippov_1988}, 
there exists a measurable selection $\sigma(t) \in \operatorname{SGN}(\xi_f(t))$ along the Filippov solution. 
Combining this with the system model \eqref{eq:dym} and control law \eqref{eq:compact_control_law} yields the following closed-loop dynamics in the Filippov sense:

\begin{equation}\label{eq:xi_f_dym}
\begin{aligned}
M_f\dot\xi_f
&=u_f-C_fv_f-G_f-M_f\dot\eta_f \\
&=\dot\eta_f
-C_f\xi_f
-Y_f\tilde\Theta_f
-\beta {\,} \sigma(t).
\end{aligned}
\end{equation}

The derivation follows directly from \eqref{eq:V2dot} and is omitted for brevity.
Since $\dot\eta_f$, $M_f^{-1}$, $C_f$, $Y_f$, $\xi_f$, and $\tilde\Theta_f$ are bounded on the compact set under consideration, and $\sigma(t) \in [-1,1]^{dn_f}$, it follows that $\dot\xi_f$ is essentially bounded. Moreover, recalling that  $\dot{\tilde\eta}_f$ is bounded, $\dot{\tilde v}_f=\dot\xi_f+\dot{\tilde\eta}_f$ is essentially bounded.

Since $\dot\xi_f$ and $\dot {\tilde v}_f$ are essentially bounded on $[0,\infty)$, $\xi_f$ and $\tilde v_f$ are globally Lipschitz continuous, which implies that $t\mapsto\Vert{}\xi_f(t)\Vert{}$ and $t\mapsto\Vert{}\tilde v_f(t)\Vert{}$ are uniformly continuous on $[0,\infty)$. Combining this with $\xi_f, \tilde v_f \in \mathcal L_1$ and applying Barbalat's lemma \cite[Lemma~8.2]{Khalil_nonlinear_2002}, one obtains $\lim_{t\to\infty}\xi_f(t)=0$ and $\lim_{t\to\infty}\tilde v_f(t)=0$.
Moreover, noting that $\tilde v_f=\xi_f+\tilde\eta_f$, one obtains $\lim_{t\to\infty}\tilde\eta_f(t)=0$. Consequently, every $\omega$-limit point of the closed-loop system belongs to
\begin{equation}
\mathcal S_0
=
\left\{
(\tilde p_f,\xi_f,\tilde\eta_f,\tilde\Theta_f)
\,\middle|\,
\dot g=0,\,
\xi_f=0,\,
\tilde v_f=0,\,
\tilde\eta_f=0
\right\}.
\end{equation}

It remains to analyze the position error on $\mathcal S_0$. Since the preceding analysis establishes that $\tilde v_f \in \mathcal L_1$, the relation $\dot{\tilde p}_f = \tilde v_f$ implies that $\tilde p_f(t)$ converges to a finite limit. Specifically, $\lim_{t\to\infty}\tilde p_f(t)=\tilde p_{f,\infty}$, where $\tilde p_{f,\infty} =\tilde p_f(0)+\int_0^\infty \tilde v_f(s)\,ds$ is a constant vector. Moreover, it follows from $e(t)-e^*=\bar R\tilde p_f(t)$ that $e(t)$ tends to a finite limit $ e_\infty=e^*+\bar R\tilde p_{f,\infty}$, which indicates that $g(t)$ converges to a finite limit $g_\infty$. Consequently the subsequent analysis is reduced to directly bounding the limiting residual $\tilde p_{f,\infty}$.

By standard measurable selection arguments in \cite{Filippov_1988}, there exists a measurable selection
$\zeta(s)\in\operatorname{SGN}(\dot g(s))$ along the Filippov solution, such that $\dot{\tilde\eta}_f=-k_p\bar R^{T}(g-g^{*})-k_v\bar R^{T}\dot g
-\gamma\bar R^{T}\operatorname{diag}(P_{g_k})\zeta-\mathbf 1_{n_f}\otimes\dot v_r$.
The auxiliary dynamics in \eqref{eq:compact_control_law_eta} are governed by  $\dot\eta_f=\phi_f$ 
almost  everywhere, where $\phi_f=-k_p\bar R^{T}(g-g^{*})-k_v\bar R^{T}\dot g
-\gamma\bar R^{T}\operatorname{diag}(P_{g_k})\zeta$.
It then follows that $\dot{\tilde\eta}_f=\phi_f-\mathbf 1_{n_f}\otimes\dot v_{r}$. Thus, for any $T>0$, one has
\begin{equation}\label{integral_dissipation}
\begin{aligned}
&\frac{\tilde{\eta}_f(t+T)-\tilde{\eta}_f(t)}{T}
=
\frac{1}{T}\int_t^{t+T}\dot{\tilde{\eta}}_f(s)\,ds \\
&\qquad\quad
=
\frac{1}{T}\int_t^{t+T}\phi_f(s)\,ds
-
\mathbf{1}_{n_f}\otimes
\frac{1}{T}\int_t^{t+T}\dot{v}_r(s)\,ds .
\end{aligned}
\end{equation}

As $t \to \infty$, the left-hand side of \eqref{integral_dissipation} vanishes since $\tilde\eta_f \to 0$. The proportional term $k_p\bar R^{T}(g-g^{*})$ in $\phi_f$ converges to $k_p\bar R^{T}(g_\infty-g^{*})$; 
the average of the $k_v\bar R^{T}\dot g$  term in $\phi_f$, given by $k_v\bar R^{T}T^{-1}[g(t{+}T)-g(t)]$, tends to zero since $g(t)\to g_\infty$; 
the average acceleration satisfies $\|T^{-1}\int\dot v_r\|\le\delta_a$ by Assumption~\ref{assum:acc-bound}. 
Moreover, since $\zeta(s)\in[-1,1]^{md}$ and $\|\operatorname{diag}(P_{g_k})\|\le1$,
the robust term satisfies the uniform bound $\|\gamma T^{-1}\!\int_{t}^{t+T}\bar R^{T}\operatorname{diag}(P_{g_k})\zeta\,ds\|
\le c_g\gamma$ for all $t$,
where $c_g=\|\bar R\|\sqrt{md}$. 
Rearranging \eqref{integral_dissipation}, taking norms on both sides and letting 
$t\to\infty$, one obtains the steady-state estimate:
\begin{equation}\label{eq:steady-state-balance}
\bigl\|k_p\bar R^{T}(g_\infty-g^{*})\bigr\|\le c_g\gamma+\sqrt{n_f}\delta_a .
\end{equation}

Taking the inner product of $k_p\bar R^{T}(g_\infty-g^{*})$ with 
$\tilde p_{f,\infty}$ and invoking $\bar R\tilde p_{f,\infty}=e_\infty-e^{*}$ 
together with the Cauchy--Schwarz inequality, \eqref{eq:steady-state-balance} implies that
$k_p(e_\infty-e^{*})^{T}(g_\infty-g^{*})\le
(c_g\gamma+\sqrt{n_f}\delta_a)\|\tilde p_{f,\infty}\|$.

Combining Lemma~\ref{lem:ZHAO_positive_lemma} with the relation $(e-e^{*})^{T}(g-g^{*})=\sum_k(\|e_k\|+\|e_k^{*}\|)(1-g_k^{T}g_k^{*})\ge e^{T}(g-g^{*})$, and noting that $\max_k\|e_{k,\infty}\|\le\bar e:=\max_k(\|e_k^{*}\|+\mu\|N_{k}\|)$ holds on $\Omega_c$, one has
\begin{equation}
(e_\infty-e^{*})^{T}(g_\infty-g^{*})\ge a_0\|\tilde p_{f,\infty}\|^2,{\,}
a_0=\frac{\lambda_{\min}(\mathcal B_{ff}^{*})}{2\bar e}.
\end{equation}

From the preceding analysis, one obtains $a_0 k_{p}\|\tilde p_{f,\infty}\|^2 \le (\sqrt{n_f}\delta_a+c_g\gamma)\|\tilde p_{f,\infty}\|$. 
If $\tilde p_{f,\infty}=0$, this inequality holds trivially; if $\tilde p_{f,\infty}\ne0$, one has $\|\tilde p_{f,\infty}\| \le a_0^{-1} k_p^{-1}(\sqrt{n_f}\delta_a+c_g\gamma)$. 
We thus conclude that 
\begin{equation}\label{pe_boundedness}
\limsup_{t\to\infty}
\|\tilde p_{f}(t)\|\le \frac{2\bar e\,(\sqrt{n_f}\delta_a+c_g\gamma)}{k_p\,\lambda_{\min}(\mathcal B_{ff}^{*})}.
\end{equation}

Consequently, the follower position tracking error is uniformly ultimately bounded (UUB), with its ultimate bound reducible via the control gain $k_p$, as indicated in \eqref{pe_boundedness}. 
In addition, the preceding analysis invoking Barbalat's lemma guarantees that the velocity tracking error converges asymptotically, i.e., $\lim_{t\to\infty}\tilde v_f(t)=0$. 
Therefore, we conclude that the closed-loop system achieves local practical formation tracking. This completes the proof.  \hfill $\blacksquare$

The constant-velocity formation tracking problem, which is an open problem itself, can be naturally regarded as a special case of the time-varying problem addressed in Theorem~\ref{thm:main}.
In this special case, a stronger result is obtained: the practical tracking in Theorem~\ref{thm:main} is elevated to asymptotic tracking, and neither the initial-condition restriction \eqref{eq:initial-condition} nor the rigidity preservation in Lemma~\ref{lem:local_preservation} is required. 

Nevertheless, without the initial-condition restriction, the invariance of $\Omega_c$ can no longer be used to ensure  inter-agent collision avoidance. Instead, we introduce the following standard assumption.
 
\begin{assum}\label{assum:separation}
Assume that no agents collide with each other during formation evolution, i.e., $\|e_k\|>\underline e$  for all $t \ge 0$.
\end{assum}

Under Assumption~\ref{assum:separation}, Theorem~\ref{thm:main} yields the following corollary for the constant-velocity formation tracking problem.

\begin{cor}\label{cor:constant_velocity}
Consider the communication-free multi-agent system \eqref{eq:dym} under Assumption~\ref{assum:positive-definite} and \ref{assum:separation}, 
and suppose the leaders move with a constant velocity, i.e., $\dot v_r(t)\equiv 0$, such that
Assumptions~\ref{assum:acc-bound} holds with $\delta_a=0$. 
Let the followers be driven by the control law \eqref{eq:local_control_law} with the design parameters $\gamma=0$ and $\beta>0$, using only the bearings $\{g_{ij},\dot g_{ij}\}_{j\in\mathcal N_i}$ and their own states $\{p_{i} ,v_i\}$. 
Then, the follower position and velocity tracking errors converge asymptotically to zero, i.e., 
\begin{equation}\label{eq:cor_conclusion}
\lim_{t\to\infty}\tilde v_f(t)=0,
\qquad
\lim_{t\to\infty}\tilde p_f(t)=0 .
\end{equation}

In particular, neither the initial condition in \eqref{eq:initial-condition} nor the rigidity-preservation requirement in Lemma~\ref{lem:local_preservation} is needed, and the system achieves asymptotic formation tracking. 
\end{cor}

\textit{Proof: }
Following a similar argument to the derivation of \eqref{eq:Vdot-bound} in the proof of Theorem~\ref{thm:main}, together with $\delta_a=0$ and $\gamma=0$, one has
\begin{equation}\label{eq:cor_Vdot}
\begin{aligned}
\dot V
\le{}&
-k_v\,\dot e^{T}\operatorname{diag}\left(\frac{P_{g_k}}{\|e_k\|}\right)\dot e
-\beta\|\xi_{f}\|\\
\le {}& 
-k_v\sum_{k=1}^{m}\|e_k\|\,\|\dot g_k\|^{2}
-\beta\|\xi_{f}\|\le 0,
\end{aligned}
\end{equation}
where the second inequality follows from $\dot g_k=P_{g_k}\dot e_k/\|e_k\|$ and $P_{g_k}^{2}=P_{g_k}$. 
Unlike \eqref{eq:Vdot-cond}, which is conditional on $\Omega_c$,  the inequality \eqref{eq:cor_Vdot} holds for almost all $t \ge0$.
Hence, $V(t)\le V(0)$, and $\xi_f$, $\tilde\eta_f$, $\tilde\Theta_f$, and $e^{T}(g-g^{*})$ are bounded. 
Since $\|e_k\|\le\|e_k^{*}\|+\sqrt2\|\tilde p_f\|$, Lemma~\ref{lem:ZHAO_positive_lemma} gives
$e^{T}(g-g^{*})\ge\lambda_{\min}(\mathcal B_{ff}^{*})\|\tilde p_f\|^{2}/\bigl(2(\max_k\|e_k^{*}\|+\sqrt2\|\tilde p_f\|)\bigr)$,
whose right-hand side is unbounded in $\|\tilde p_f\|$.
Thus, the boundedness of $e^{T}(g-g^{*})$ implies that there exists $\bar p>0$ such that $\|\tilde p_f(t)\|\le\bar p$ and $\|e_k(t)\|\le\bar e':=\max_k\|e_k^{*}\|+\sqrt2\,\bar p$.
Moreover, $\tilde v_f = \xi_f + \tilde\eta_f$ is bounded, which, together with Assumption~\ref{assum:separation}, yields that $\dot g_k = P_{g_k}(r_k^{T}\otimes I_d)\tilde v_f/\Vert{}e_k\Vert{}$ is bounded. 
Consequently,  $\dot\eta_f$ is also bounded by \eqref{eq:compact_control_law_eta}.

By combining \eqref{eq:cor_Vdot} and Assumption~\ref{assum:separation}, one obtains $\xi_f\in\mathcal L_1$ and  $\dot g\in\mathcal L_2$.
Since $\dot\eta_f$, $M_f^{-1}$, $C_f$, $Y_f$, $\xi_f$, and $\tilde\Theta_f$ are bounded on the compact set under consideration, the dynamics of $\xi_f$ \eqref{eq:xi_f_dym} implies that $\dot\xi_f$ is bounded almost everywhere.
Hence $\|\xi_f\|$ is uniformly continuous, and via Barbalat's lemma \cite[Lemma~8.2]{Khalil_nonlinear_2002} one obtains that $\lim_{t\to\infty}\xi_f(t)=0$.

Let $w=\bar R^{T}(g-g^{*})$, which is bounded and satisfies $\dot w=\bar R^{T}\dot g$.
Substituting $\dot v_r=0$ and $\gamma=0$ into \eqref{eq:compact_control_law_eta} yields $\dot{\tilde\eta}_f=-k_p\,w-k_v\bar R^{T}\dot g$. 
Combining this with the identity $\bar R\tilde\eta_f=\bar R\tilde v_f-\bar R\xi_f=\dot e-\bar R\xi_f$, it follows that
\begin{equation}\label{eq:cor_cross}
\frac{d}{dt}\bigl(\tilde\eta_f^{T}w\bigr)
=-k_p\|w\|^{2}-k_v w^{T}\bar R^{T}\dot g+\dot e^{T}\dot g-\xi_f^{T}\bar R^{T}\dot g .
\end{equation}
By Young's inequality, $-k_v w^{T}\bar R^{T}\dot g\le\frac{k_p}{2}\|w\|^{2}+\frac{k_v^{2}\|\bar R\|^{2}}{2k_p}\|\dot g\|^{2}$; moreover, $\dot e^{T}\dot g=\sum_{k=1}^{m}\|e_k\|\,\|\dot g_k\|^{2}\le\bar e'\,\|\dot g\|^{2}$ and $|\xi_f^{T}\bar R^{T}\dot g|\le\|\bar R\|\,\bar g\,\|\xi_f\|$, where $\bar g=\sup_{t\ge0}\|\dot g(t)\|<\infty$ as $\dot g_k$ is bounded. 
Integrating \eqref{eq:cor_cross} over $[0,t]$ yields 
\begin{equation}\label{eq:cor_w_ineq}
\begin{aligned}
\frac{k_p}{2}\int_{0}^{t}\|w\|^{2}d\tau
\le{}&
2\sup_{\tau\ge0}|\tilde\eta_f^{T}(\tau)w(\tau)|
\\
{}&+c_1\int_{0}^{\infty}\|\dot g\|^{2}d\tau+c_2\int_{0}^{\infty}\|\xi_f\|\,d\tau,
\end{aligned}
\end{equation}
where $c_1=\frac{k_v^{2}\|\bar R\|^{2}}{2k_p}+\bar e'$ and $c_2=\|\bar R\|\,\bar g$. Given the boundedness of $\tilde{\eta}_{f}$ and $w$, together with $\xi_f \in \mathcal{L}_1$ and $\dot{g} \in \mathcal{L}_2$, the right-hand side of \eqref{eq:cor_w_ineq} is bounded, which implies that $w \in \mathcal{L}_2$.
Furthermore, since $\frac{d}{dt}\|w\|^{2}=2w^{T}\bar R^{T}\dot g$ is bounded, $\|w\|^{2}$ is uniformly continuous. By Barbalat's lemma \cite[Lemma~8.2]{Khalil_nonlinear_2002}, one has $\lim_{t\to\infty}w(t)=\lim_{t\to\infty}\bar R^{T}(g-g^{*})=0$.
 
Since $\bar R\tilde p_f=e-e^{*}$, by combining Lemma~\ref{lem:ZHAO_positive_lemma} with the inequality $(e-e^{*})^{T}(g-g^{*})\ge e^{T}(g-g^{*})$ established in the proof of Theorem~\ref{thm:main}, one has
\begin{equation}\label{eq:cor_final}
\begin{aligned}
\frac{\lambda_{\min}(\mathcal B_{ff}^{*})}{2\bar e'}\,\|\tilde p_f\|^{2}
\le{}& e^{T}(g-g^{*})
\le (e-e^{*})^{T}(g-g^{*})\\
={}&
\tilde p_f^{T}w
\le \bar p\,\|\bar R^{T}(g-g^{*})\| ,
\end{aligned}
\end{equation}
whose right-hand side tends to zero, yielding $\lim_{t\to\infty}\tilde p_f(t)=0$.
Consequently, $\int_{0}^{t}\tilde\eta_f\,d\tau=\tilde p_f(t)-\tilde p_f(0)-\int_{0}^{t}\xi_f\,d\tau$ converges to a finite limit as $t\to\infty$, yielding $\tilde\eta_f\in\mathcal L_1$.
Since $\dot{\tilde\eta}_f$ is bounded, $\tilde\eta_f$ is uniformly continuous, and via Barbalat's lemma \cite[Lemma~8.2]{Khalil_nonlinear_2002} one has $\lim_{t\to\infty}\tilde\eta_f(t)=0$. 
Together with $\lim_{t\to\infty}\xi_f(t)=0$, it follows that $\lim_{t\to\infty}\tilde v_f(t)=0$, which completes the proof. \hfill $\blacksquare$

\section{Simulation Results}\label{Section_4}

\begin{figure*}[!t]
\centering

\subfloat[]{%
\includegraphics[width=0.49\textwidth]{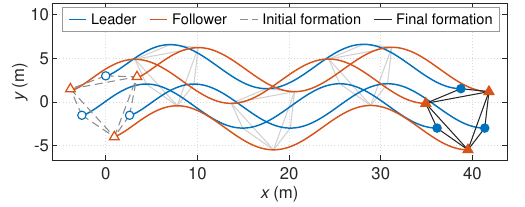}
\label{fig:sim_1a}}
\hfil
\subfloat[]{%
\includegraphics[width=0.49\textwidth]{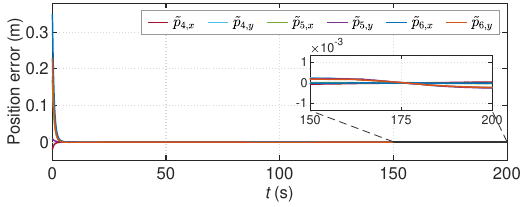}
\label{fig:sim_1b}}

\vspace{1mm}

\subfloat[]{%
\includegraphics[width=0.49\textwidth]{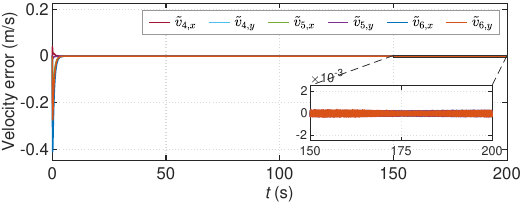}
\label{fig:sim_1c}}
\hfil
\subfloat[]{%
\includegraphics[width=0.49\textwidth]{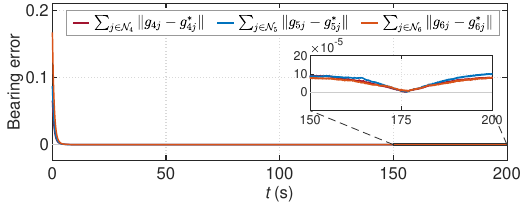}
\label{fig:sim_1d}}

\caption{Formation tracking under a time-varying leader velocity: 
(a) formation trajectories; 
(b) follower position errors $\tilde p_{i}$; 
(c) follower velocity errors $\tilde v_{i}$; 
(d) follower bearing errors $\sum_{j\in\mathcal N_i}\|g_{ij}-g_{ij}^{*}\|$, $i=4,5,6$.}
\label{fig:sim_1}
\end{figure*}

\begin{figure*}[!t]
\centering

\subfloat[]{%
\includegraphics[width=0.49\textwidth]{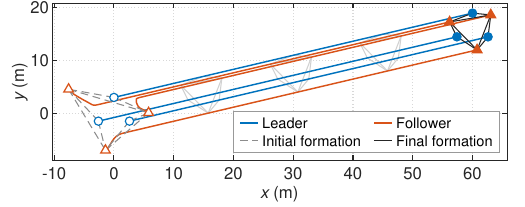}
\label{fig:sim_2a}}
\hfil
\subfloat[]{%
\includegraphics[width=0.49\textwidth]{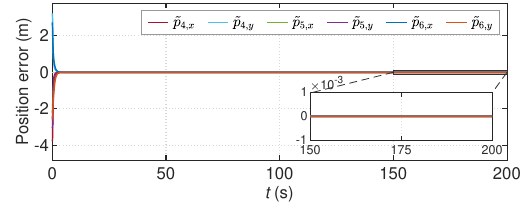}
\label{fig:sim_2b}}

\vspace{1mm}

\subfloat[]{%
\includegraphics[width=0.49\textwidth]{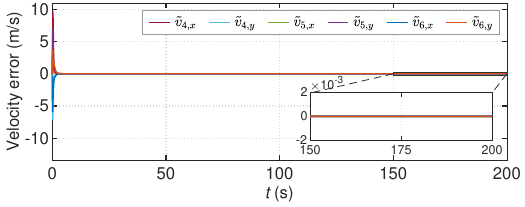}
\label{fig:sim_2c}}
\hfil
\subfloat[]{%
\includegraphics[width=0.49\textwidth]{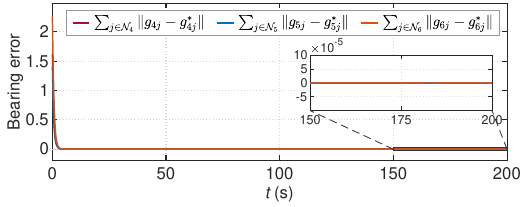}
\label{fig:sim_2d}}

\caption{Formation tracking under a constant leader velocity: 
(a) formation trajectories; 
(b) follower position errors $\tilde p_{i}$; 
(c) follower velocity errors $\tilde v_{i}$; 
(d) follower bearing errors $\sum_{j\in\mathcal N_i}\|g_{ij}-g_{ij}^{*}\|$, $i=4,5,6$.}
\label{fig:sim_2}
\end{figure*}

To validate the proposed control law \eqref{eq:local_control_law}, we consider a
planar MAS consisting of $n=6$ agents, of which $n_l=3$ are leaders, indexed by
$\{1,2,3\}$, and $n_f=3$ are followers, indexed by $\{4,5,6\}$. The formation configuration
described in this paragraph is common to the two scenarios of constant and time-varying leader velocities. 
Following the simulation setup in \cite{Zhang_Auto_2017_Simulation_Model}, the dynamics of
each follower $i$ are governed by \eqref{eq:dym}, with
$M_i=[\vartheta_1+2\vartheta_2 c_{iy},\,\vartheta_3+\vartheta_2 c_{iy};\,
\vartheta_3+\vartheta_2 c_{iy},\,\vartheta_3]$, 
$C_i=[-\vartheta_2 s_{iy}v_{iy},\,-\vartheta_2 s_{iy}(v_{ix}+v_{iy});\,
\vartheta_2 s_{iy}v_{ix},\,0]$, and $G_i=0$, where $c_{iy}=\cos p_{iy}$ and
$s_{iy}=\sin p_{iy}$. The unknown parameters are
$\Theta_i=[\vartheta_1,\vartheta_2,\vartheta_3]^{T}=[1.301,0.056,0.296]^{T}$. 
The graph has edge set $\mathcal E = \{ (1,4),$ $(2,5), $ $(3,6), $ $(2,4), $ $(3,5),$ $(1,6), $ $(4,5), $ $(5,6), $ $(4,6) \}$, and the desired bearings $g^{*}_{ij}$ are induced through \eqref{eq:e_ij_g_ij} by the target configuration, whose positions at $t=0$ are $p^{*}_1=(0,3)$, $p^{*}_2=(-2.6,-1.5)$, $p^{*}_3=(2.6,-1.5)$, $p^{*}_4=(-3.83,1.32)$, $p^{*}_5=(0.77,-3.98)$ and $p^{*}_6=(3.06,2.66)$. 
In both scenarios the leaders start on target, $p_i(0)=p^{*}_i$ for $i\in\{1,2,3\}$; 
the initial velocities, auxiliary states, and estimator states of followers are initialized to zero; 
and the control gains and design parameters are chosen as $k_p=40$, $k_v=30$, $\beta=0.1$ and $\Lambda_i=45I_3$.

\subsection{Case 1: Leaders with Time-Varying Velocity}
This scenario is presented to validate the theoretical results in Theorem~\ref{thm:main}. The leaders move with the time-varying velocity $v_r(t)=[\,v_0\tanh(t/\tau),\,v_0\kappa\tanh(t/\tau)\sin(\omega t)\,]^{T}~
\mathrm{m/s}$, with $v_0=0.212$, $\kappa=0.75$, $\omega=2\pi/100$, and $\tau=25$, and the design parameter $\gamma$ is set to $\gamma=1$. 
The followers are initially placed at $p_4(0)=(-3.85,1.5)$, $p_5(0)=(0.94,-3.97)$ and $p_6(0)=(3.41,2.89)$. 
For $\mu=0.93\in(0,\bar\mu)$, these choices yield $\underline\lambda(\mu)=0.065$ and $V(0)=1.31<\bar V(\mu)=2.13$, such that
conditions \eqref{eq:initial-condition}--\eqref{eq:xi_condition} of
Theorem~\ref{thm:main} are satisfied.

Fig.~\ref{fig:sim_1} presents the simulation results of this scenario. As shown in Fig.~\ref{fig:sim_1}(a), the followers successfully achieve the prescribed geometry and subsequently track the leaders moving with time-varying velocity. 
The position errors in Fig.~\ref{fig:sim_1}(b) converge within a few seconds and remain within a small neighborhood of zero (see the inset), validating the UUB position tracking established in Theorem~\ref{thm:main}. 
Meanwhile, the velocity errors in Fig.~\ref{fig:sim_1}(c) converge to zero, where the chattering is inherently induced by the discontinuous control terms. 
The bearing errors in Fig.~\ref{fig:sim_1}(d) converge to a small neighborhood of zero accordingly. 
These observations are consistent with the local practical formation tracking result asserted in Theorem~\ref{thm:main}.

\subsection{Case 2: Leaders with Constant Velocity}
This scenario is given to verify Corollary~\ref{cor:constant_velocity} with the same setup as in Case 1  except for two modifications.  
First, the leaders move at a constant velocity $v_r=[0.3,\,0.08]^{T}~\mathrm{m/s}$, and the
design parameter $\gamma$ is set to $\gamma=0$ in accordance with Corollary~\ref{cor:constant_velocity}.
Second, to exhibit the global stability, the initial positions of the followers are randomly generated within a disk of radius $4$ centered at their respective target locations $p^{*}_i$, $i\in\{4,5,6\}$. 
These positions typically lie far outside the admissible set $\Omega_c$ in Theorem~\ref{thm:main}, and the initial-condition restriction \eqref{eq:initial-condition} is completely removed. 

Fig.~\ref{fig:sim_2} presents the simulation results of this scenario. 
As illustrated in Fig.~\ref{fig:sim_2}(a), the followers successfully achieve the desired formation and track the constant-velocity leaders, despite large initial errors that lie outside the admissible set $\Omega_c$ of Theorem~\ref{thm:main}. This demonstrates the formation tracking asserted in Corollary~\ref{cor:constant_velocity}.
More notably, as in Fig.~\ref{fig:sim_2}(b), the position errors converge to zero without any residual, thereby eliminating the ultimate bound observed in the time-varying case (Fig.~\ref{fig:sim_1}(b)). 
The velocity errors in Fig.~\ref{fig:sim_2}(c) and the bearing errors in Fig.~\ref{fig:sim_2}(d)
both converge to zero. 
These results validate the asymptotic tracking result established in Corollary~\ref{cor:constant_velocity}.

\section{Conclusion} \label{Section_5}
This paper has investigated communication-free bearing-only formation tracking for multi-agent systems with followers of uncertain Euler--Lagrange dynamics and leaders of time-varying velocities. 
A distributed adaptive control law is developed using only local bearing measurements, without requiring leaders' state information or inter-agent communication. 
By exploiting bearing-rate measurements to extract partial relative-velocity signals, in conjunction with bearing rigidity properties, the control scheme establishes a robust damping mechanism to compensate for the unavailable velocity error. 
Furthermore, a purely local, bearing-driven auxiliary variable is introduced, which, incorporating this damping mechanism, provides a surrogate velocity error to facilitate the adaptive control design for EL dynamics. 
Through a Filippov-based Lyapunov analysis, the closed-loop system is shown to achieve local practical formation tracking over an explicit initial-condition set derived from a bearing-rigidity preservation argument. 
The proposed law can subsume stationary and constant-velocity formations as special cases. Future work will focus on relaxing the restriction on the initial-condition set. 

\appendices

\section{Proof of Lemma \ref{lem:local_preservation}} \label{Appendix:proof_local}

For notational brevity, throughout this proof we denote 
$\alpha_k = \|(g_k^{\ast})^{T} N_{k}\|$ and $\beta_k = \|P_{g_k^{\ast}} N_{k}\|$, 
which are constants that can be calculated \textit{a priori} based on the desired bearing vectors and the underlying graph topology. 

Let $z = (\mathcal B_{ff}^{\ast})^{1/2} \tilde p_{f}$ denote the rigidity-weighted error. Then, for any $\tilde p_{f}\in\Omega_c$, one has $\|z\|_2^{2} < \mu^{2}\Leftrightarrow\|z\|<\mu$. 
From $\tilde {p}={\rm col}(0,\tilde{p}_f)$, we obtain $e-e^{\ast}= \bar{H}(p-p^{\ast})=(R\otimes I_d)\tilde p_{f}$, which, combined with $\tilde p_{f} = (\mathcal B_{ff}^{\ast})^{-1/2} z$ yields
\begin{equation}\label{eq:edge-perturbation}
    e_k - e_k^{\ast} = (r_k^{T}\otimes I_d)\,\tilde p_{f} = N_{k} z,{\,}{\,}k = 1,\ldots,m.
\end{equation}

Decomposing $N_{k} z$ along $g_k^{\ast}$ and its orthogonal complement yields  
\begin{equation}\label{eq:decomp}
N_{k} z
=
\underbrace{(g_k^{*})^{T}N_{k} z}_{\textstyle a_k}{\,}g_k^{*}
+
\underbrace{P_{g_k^{*}}N_{k} z}_{\textstyle b_k}.
\end{equation}

The Cauchy--Schwarz inequality together with the definitions of $\alpha_k$ and $\beta_k$ yields 
$|a_k| \le \alpha_k \|z\| < \alpha_k \mu$ and 
$\|b_k\| \le \beta_k \|z\| < \beta_k\mu$. 
Since $\alpha_k\mu < \|e_k^{\ast}\|$ as stated in Lemma~\ref{lem:local_preservation}, it follows that
$\|e_k^{\ast}\| + a_k \ge \|e_k^{\ast}\| - |a_k| > \|e_k^{\ast}\| - \alpha_k\mu > 0$.
Therefore,  the edge vector 
$e_k = e_k^{\ast} + N_{k} z = (\|e_k^{\ast}\| + a_k) g_k^{\ast} + b_k$ 
is non-zero since 
$\|e_k\|_2^{2} = (\|e_k^{\ast}\| + a_k)^{2} + \|b_k\|_2^{2} \ge (\|e_k^{\ast}\| + a_k)^{2}$.

Let $\theta_k$ denote the angle between $g_k $ and $g_k^{\ast}$. 
Since $b_k \perp g_k^{\ast}$ and $\|e_k^{\ast}\| + a_k > 0$, one has $\sin\theta_k = \|b_k\|/\|e_k\|$, which together with the above lower bound on $\|e_k\|$ yields
\begin{equation}\label{eq:sintheta}
    \sin\theta_k 
    \le \frac{\|b_k\|}{\|e_k^{\ast}\| + a_k} 
    \le \frac{\beta_k\mu}{\|e_k^{\ast}\| - \alpha_k\mu} 
    = \varepsilon_k(\mu).
\end{equation}

For any pair of unit vectors $u, v \in \mathbb{R}^{d}$, one has 
$\|P_u - P_v\|_2 = \sin\angle(u, v)$ \cite[Sec.\ 5.3]{Stewart_Sun_1990}. 
Applied to $g_k$ and $g_k^{\ast}$, this identity and \eqref{eq:sintheta} imply 
\begin{equation}\label{eq:projection-difference-bound}
\|P_{g_k} - P_{g_k^{\ast}}\|_2 \le \varepsilon_k(\mu),
{\,}k = 1,\ldots,m.
\end{equation}

Moreover, it follows from $\mathcal B_{ff}=\bar R^T\operatorname{diag}(P_{g_k})\bar R$ that 
\begin{equation}\label{eq:Bff-diff-decomp}
\begin{aligned}
\mathcal B_{ff} - \mathcal B_{ff}^{\ast}
={}&
(R\otimes I_d)^T
\bigl[
\operatorname{diag}(P_{g_k})
-\operatorname{diag}(P_{g_k^{\ast}})
\bigr]
(R\otimes I_d) \\
={}&
\sum_{k=1}^{m}
\bigl(r_k r_k^{T}\bigr)
\otimes
\bigl(P_{g_k} - P_{g_k^{\ast}}\bigr).
\end{aligned}
\end{equation}

For any $x \in \mathbb{R}^{dn_f}$, define $y_k = (r_k^{T}\otimes I_d)\,x \in \mathbb{R}^{d}$. 
Then \eqref{eq:Bff-diff-decomp} together with \eqref{eq:projection-difference-bound} yields
\begin{equation}\label{eq:quad-bound} 
\begin{aligned} 
\bigl|x^{T}(\mathcal B_{ff} - \mathcal B_{ff}^{\ast})\,x\bigr| 
={}& 
\Bigl| \sum_{k=1}^{m} y_k^{T}(P_{g_k} - P_{g_k^{\ast}})\,y_k \Bigr| \\ 
\le{}& 
\sum_{k=1}^{m} \varepsilon_k(\mu)\,\|y_k\|_2^{\,2} \\ 
={}& 
x^{T}\bigl(L_{\varepsilon}(\mu)\otimes I_d\bigr)\,x . 
\end{aligned} 
\end{equation}

According to \eqref{eq:quad-bound}, one has 
\begin{equation}\label{eq:Bff-lower-bound}
\begin{aligned}
x^{T}\mathcal B_{ff}x
\ge{}&
x^{T}\mathcal B_{ff}^{\ast}x
-
x^{T}\bigl(L_{\varepsilon}(\mu)\otimes I_d\bigr)x  \\
\ge{}&
\bigl(
\lambda_{\min}(\mathcal B_{ff}^{\ast})
-
\lambda_{\max}(L_{\varepsilon}(\mu))
\bigr)\|x\|^{2},
\end{aligned}
\end{equation}
which implies $x^T\mathcal B_{ff}x\ge \underline\lambda(\mu)\|x\|^2, \forall x\in\mathbb R^{dn_f}$ by the definition of \(\underline\lambda(\mu)\) in \eqref{eq:Bff_inequality}. 
Therefore, \(\lambda_{\min}(\mathcal B_{ff}(p))\ge \underline\lambda(\mu)\) for all \(\tilde p_{f}\in\Omega_c\). 

Moreover, for each $k$, the function $\mu\mapsto \varepsilon_k(\mu)$ in~\eqref{eq:def_epsilon} is continuous and $\varepsilon_k(0) = 0$. Hence $L_\varepsilon(\mu)$ is continuous in $\mu$ in any matrix norm, and so is $\lambda_{\max}(L_\varepsilon(\mu))$ by Weyl's continuity of eigenvalues of symmetric matrices. Therefore, $\underline{\lambda}(\mu)$ is continuous in $\mu$ with $\lim_{\mu\to 0^{+}}\underline{\lambda}(\mu) = \lambda_{\min}(\mathcal B_{ff}^{\ast}) > 0$. The existence of $\bar\mu>0$ such that $\underline\lambda(\mu)>0$ on $(0,\bar\mu)$ follows. The proof ends here. \hfill $\blacksquare$

\section{Proof of Proposition~\ref{prop:Omega_invariant}}\label{Appendix:proof_invariant}

Define $s(t) = \tilde p_{f}^{T}(t)\,\mathcal{B}_{ff}^{\ast}\,\tilde p_{f}(t)$ such that $\tilde p_{f}(t)\in\Omega_c \Leftrightarrow s(t)<\mu^{2}$. According to \eqref{eq:edge-perturbation}, one has
\begin{equation}\label{eq:e_k_inequality}
\|e_k(t)\| \leq \|e_k^{\ast}\|+\|N_{k} z(t)\| \leq \|e_k^{\ast}\|+\|N_{k}\|\sqrt{s(t)},
\end{equation}
where the last step comes from $z(t) = (\mathcal{B}_{ff}^{\ast})^{1/2}\,\tilde p_{f}(t)$ and $\|z(t)\|_2^{2} = s(t)$. 
Combining \eqref{eq:e_k_inequality} with Lemma~\ref{lem:ZHAO_positive_lemma}, one has
\begin{equation}\label{eq:V-lower}
\begin{aligned}
V(t)
&\ge k_p\,e^{T}(g-g^{*}) \ge
\frac{k_p\tilde p_{f}^{T}\mathcal{B}_{ff}^{*}\tilde p_{f}}{2\max_k\|e_k\|}\\
&\ge
\frac{
k_p\,s(t)
}{
2\,\max\nolimits_{k}
\bigl(\|e_k^\ast\|+\|N_{k}\|\sqrt{s(t)}\bigr)
}
=: \Phi\bigl(s(t)\bigr).
\end{aligned}
\end{equation}

The map $\Phi(s(t))$ is strictly increasing in $s(t)$ and satisfies $\Phi(s(t)) = \bar V(\mu)$ as $s(t)=\mu^{2}$. Thus, the initial condition $V(0) < \bar V(\mu)$ implies $s(0) < \mu^{2}$, that is, $\tilde p_{f}(0) \in \Omega_c$. 

Suppose, for contradiction, that the trajectory leaves $\Omega_c$. Continuity of $s(\cdot)$ along the continuous Filippov solution ensures the existence of the first exit time $t_{\ast}:= \inf\{t>0:s(t) \ge \mu^{2}\}$ with $s(t_{\ast}) = \mu^{2}$ and $s(t)<\mu^{2}$ on $[0,t_{\ast})$. 
Since $s(t)<\mu^{2}$, i.e., $\tilde p_{f}(t)\in\Omega_c$, on $[0,t_{\ast})$, the conditional inequality \eqref{eq:Vdot-cond} gives
$\dot V\le 0$ almost everywhere on $[0,t_{\ast})$, such that $V$ is non-increasing on $[0,t_{\ast})$.
Since $V$ is absolutely continuous along Filippov solutions, it extends continuously to $t_{\ast}$, and 
$V(t_{\ast})=\lim_{t\to t_{\ast}^-}V(t)\le V(0)<\bar V(\mu)$.
However, substituting $s(t_{\ast}) = \mu^2$ into \eqref{eq:V-lower}, one has $V(t_{\ast}) \ge \Phi(\mu^2) = \bar V(\mu)$. 
This is a contradiction. 
Therefore, $s(t) = \tilde p_{f}^{T}(t)\mathcal{B}_{ff}^{\ast}\,\tilde p_{f}(t)< \mu^2$ for all $t \ge 0$, and thus, the set $\Omega_c$ is positively invariant. 
The proof ends here. \hfill $\blacksquare$

\bibliographystyle{IEEEtran}
\bibliography{reference_update_Auto}

\begin{thebibliography}{10}
\providecommand{\url}[1]{#1}
\csname url@samestyle\endcsname
\providecommand{\newblock}{\relax}
\providecommand{\bibinfo}[2]{#2}
\providecommand{\BIBentrySTDinterwordspacing}{\spaceskip=0pt\relax}
\providecommand{\BIBentryALTinterwordstretchfactor}{4}
\providecommand{\BIBentryALTinterwordspacing}{\spaceskip=\fontdimen2\font plus
\BIBentryALTinterwordstretchfactor\fontdimen3\font minus \fontdimen4\font\relax}
\providecommand{\BIBforeignlanguage}[2]{{%
\expandafter\ifx\csname l@#1\endcsname\relax
\typeout{** WARNING: IEEEtran.bst: No hyphenation pattern has been}%
\typeout{** loaded for the language `#1'. Using the pattern for}%
\typeout{** the default language instead.}%
\else
\language=\csname l@#1\endcsname
\fi
#2}}
\providecommand{\BIBdecl}{\relax}
\BIBdecl

\bibitem{OH_Auto_2015}
K.-K. Oh, M.-C. Park, and H.-S. Ahn, ``A survey of multi-agent formation control,'' \emph{Automatica}, vol.~53, pp. 424--440, 2015.

\bibitem{Su_ARC_2026}
H.~Su, Z.~Yang, S.~Zhu, C.~Chen, X.~Guan, and L.~Xie, ``Bearing-based multi-agent formation control: A survey and taxonomy,'' \emph{Annu. Rev. Control}, vol.~61, p. 101043, 2026.

\bibitem{Dou_IJRNC_2025}
L.~Dou and H.~Zhang, ``Moving-target localization and circumnavigation control for second-order multiagent systems with bearing-only measurements,'' \emph{Int. J. Robust Nonlinear Control}, vol.~35, no.~18, pp. 7767--7778, 2025.

\bibitem{Zhaoshiyu_TAC_2016}
S.~Zhao and D.~Zelazo, ``Bearing rigidity and almost global bearing-only formation stabilization,'' \emph{IEEE Trans. Autom. Control}, vol.~61, no.~5, pp. 1255--1268, May 2016.

\bibitem{ZhaoShiyu_Auto_2016}
------, ``Localizability and distributed protocols for bearing-based network localization in arbitrary dimensions,'' \emph{Automatica}, vol.~69, pp. 334--341, 2016.

\bibitem{TrinhMukherjee_CDC_2017}
M.~H. Trinh, D.~Mukherjee, D.~Zelazo, and H.-S. Ahn, ``Finite-time bearing-only formation control,'' in \emph{Proc. IEEE 56th Annu. Conf. Decis. Control (CDC)}, Melbourne, VIC, Australia, Dec. 2017, pp. 1578--1583.

\bibitem{TrinhZhaoSun_TAC_2019}
M.~H. Trinh, S.~Zhao, Z.~Sun, D.~Zelazo, B.~D.~O. Anderson, and H.-S. Ahn, ``Bearing-based formation control of a group of agents with leader-first follower structure,'' \emph{IEEE Trans. Autom. Control}, vol.~64, no.~2, pp. 598--613, Feb. 2019.

\bibitem{Tran_TCNS_2019}
Q.~V. Tran, M.~H. Trinh, D.~Zelazo, D.~Mukherjee, and H.-S. Ahn, ``Finite-time bearing-only formation control via distributed global orientation estimation,'' \emph{IEEE Trans. Control Netw. Syst.}, vol.~6, no.~2, pp. 702--712, 2019.

\bibitem{ZhaoShiyu_TAC_2019}
S.~Zhao, Z.~Li, and Z.~Ding, ``Bearing-only formation tracking control of multiagent systems,'' \emph{IEEE Trans. Autom. Control}, vol.~64, no.~11, pp. 4541--4554, Nov. 2019.

\bibitem{Trinh_Auto_2021}
M.~H. Trinh, Q.~V. Tran, D.~V. Vu, P.~D. Nguyen, and H.-S. Ahn, ``Robust tracking control of bearing-constrained leader--follower formation,'' \emph{Automatica}, vol. 131, p. 109733, 2021.

\bibitem{Zhaojianing_LCSS_2021}
J.~Zhao, X.~Yu, X.~Li, and H.~Wang, ``Bearing-only formation tracking control of multi-agent systems with local reference frames and constant-velocity leaders,'' \emph{IEEE Control Syst. Lett.}, vol.~5, no.~1, pp. 1--6, Jan. 2021.

\bibitem{Garanayak_TCNS_2025}
C.~Garanayak and D.~Mukherjee, ``Formation control in agents' local coordinate frames for arbitrary initial attitudes,'' \emph{IEEE Trans. Control Netw. Syst.}, vol.~12, no.~1, pp. 993--1005, 2025.

\bibitem{Song_LCSS_2024}
Z.~Song, M.~Xie, and H.~Huang, ``Bearing-only formation tracking control for multi-agent systems with time-varying velocity leaders,'' \emph{IEEE Control Syst. Lett.}, vol.~8, pp. 2027--2032, 2024.

\bibitem{Song_TIE_2026}
Z.~Song, H.~Huang, and G.~Feng, ``Bearing-only collision-free formation tracking control for multi-agent systems: A safety-critical control approach,'' \emph{IEEE Trans. Ind. Electron.}, 2026, in press, doi: 10.1109/TIE.2026.3730129.

\bibitem{ORTEGA_Auto_1989}
R.~Ortega and M.~W. Spong, ``Adaptive motion control of rigid robots: A tutorial,'' \emph{Automatica}, vol.~25, no.~6, pp. 877--888, 1989.

\bibitem{Meng_TAC_2023_EL-Dym}
X.~Meng, J.~Mei, Z.~Miao, A.~Wu, and G.~Ma, ``Fully distributed consensus of multiple {Euler--Lagrange} systems under switching directed graphs using only position measurements,'' \emph{IEEE Trans. Autom. Control}, vol.~69, no.~3, pp. 1781--1788, Mar. 2024.

\bibitem{ZhaobangweiZhaojianingYUxiao_IFAC_2023}
B.~Zhao, J.~Zhao, W.~Lan, and X.~Yu, ``Adaptive bearing-only control of multiple {Euler--Lagrange} systems for static geometric formation,'' \emph{IFAC-PapersOnLine}, vol.~56, no.~2, pp. 10\,632--10\,637, 2023, 22nd IFAC World Congress.

\bibitem{Li-Song-Xie_TCNS_2025}
K.~Li, X.~Fang, Y.~Song, and L.~Xie, ``Singularity-free task-space formation control for manipulators with joint constraints using relative bearing-only feedback,'' \emph{IEEE Trans. Control Netw. Syst.}, vol.~12, no.~4, pp. 2821--2832, 2025.

\bibitem{Li-Song_IJRNC_2025}
K.~Li, K.~Zhao, Z.~Li, and Y.~Song, ``Task-space bearing-only formation control for networked robotic manipulators without velocity measurement,'' \emph{Int. J. Robust Nonlinear Control}, vol.~35, no.~12, pp. 5227--5237, 2025.

\bibitem{Cheng_arXiv_2025}
H.~Cheng, M.~Guay, S.~Wang, and Y.~Che, ``Collision-free bearing-driven formation tracking for {Euler--Lagrange} systems,'' \emph{arXiv preprint arXiv:2508.09908}, 2025.

\bibitem{Hong_Auto_2008_Observer_Intro}
Y.~Hong, G.~Chen, and L.~Bushnell, ``Distributed observers design for leader-following control of multi-agent networks,'' \emph{Automatica}, vol.~44, no.~3, pp. 846--850, 2008.

\bibitem{Shevitz_TAC_1994}
D.~Shevitz and B.~Paden, ``Lyapunov stability theory of nonsmooth systems,'' \emph{IEEE Trans. Autom. Control}, vol.~39, no.~9, pp. 1910--1914, Sep. 1994.

\bibitem{Filippov_1988}
A.~F. Filippov, \emph{Differential Equations with Discontinuous Righthand Sides}.\hskip 1em plus 0.5em minus 0.4em\relax Dordrecht: Springer, 1988.

\bibitem{Khalil_nonlinear_2002}
H.~K. Khalil, \emph{Nonlinear Systems}, 3rd~ed.\hskip 1em plus 0.5em minus 0.4em\relax Upper Saddle River, NJ, USA: Prentice Hall, 2002.

\bibitem{Zhang_Auto_2017_Simulation_Model}
Y.~Zhang, Z.~Deng, and Y.~Hong, ``Distributed optimal coordination for multiple heterogeneous {Euler--Lagrangian} systems,'' \emph{Automatica}, vol.~79, pp. 207--213, 2017.

\bibitem{Stewart_Sun_1990}
G.~W. Stewart and J.-g. Sun, \emph{Matrix Perturbation Theory}.\hskip 1em plus 0.5em minus 0.4em\relax Boston, MA, USA: Academic Press, 1990.

\end{thebibliography}
\vspace{-10mm}

\begin{IEEEbiography}[{\includegraphics[width=1in,height=1.25in,clip,keepaspectratio]{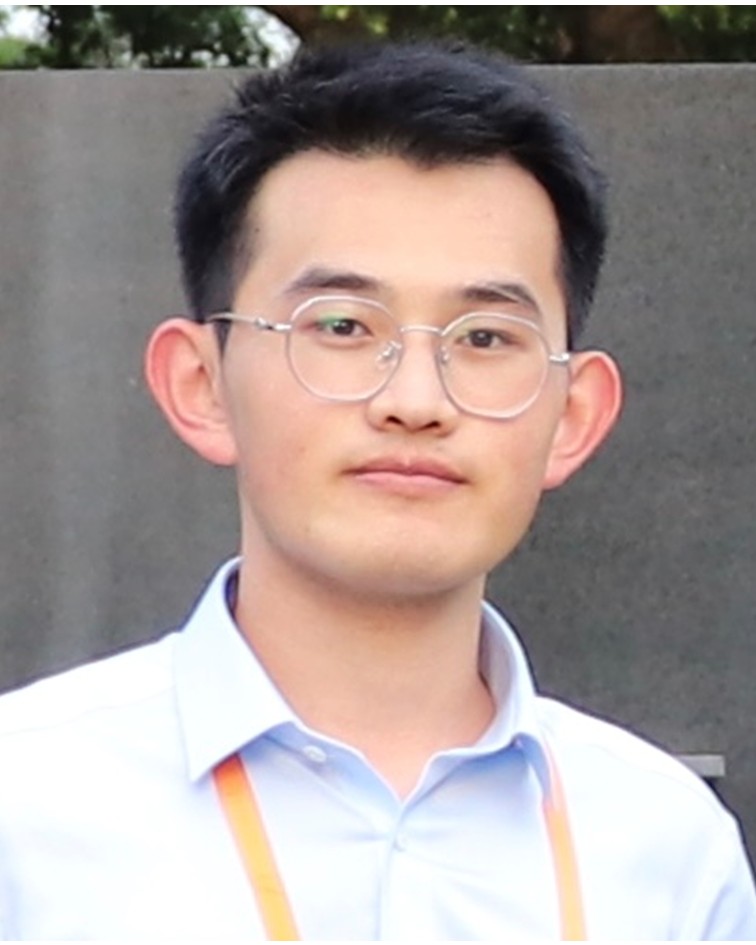}}]{Zilong Song} (Graduate Student Member, IEEE) received B.E. degree in mechanical engineering from Shandong University of Science and Technology, Shandong, China, in 2022, and M.E. degree in mechanical engineering from Zhejiang University, Zhejiang, China, in 2025.

He is currently a Ph.D. student with the Department of Mechanical Engineering, City University of Hong Kong, Hong Kong. His research interests include bearing-only formation tracking and multi-agent systems.
\vspace{-10mm}
\end{IEEEbiography}

\begin{IEEEbiography}[{\includegraphics[width=1in,height=1.25in,clip,keepaspectratio]{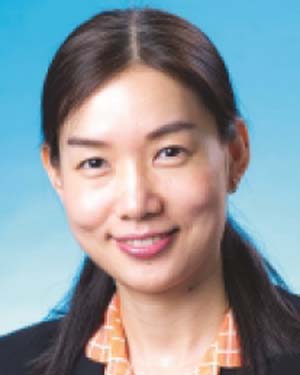}}]{Lu Liu} (Senior Member, IEEE) received the Ph.D. degree in mechanical and automation engineering in 2008 from the Department of Mechanical and Automation Engineering, Chinese University of Hong Kong, Hong Kong.

From 2009 to 2012, she was an Assistant
Professor with the University of Tokyo, Tokyo,
Japan, and then with the University of Nottingham, Nottingham, U.K. After that, she joined
City University of Hong Kong, Hong Kong,
where she is currently a Professor. Her current
research interests include networked dynamical
systems, nonlinear control systems, and multirobot systems.

Dr. Liu is an Associate Editor for \textsc {IEEE Transactions on Fuzzy Systems}, \textsc {IEEE Robotics and Automation Letters}, \textit{Control Theory and Technology}, and \textit{Unmanned Systems}.
\vspace{-10mm}
\end{IEEEbiography}

\begin{IEEEbiography}[{\includegraphics[width=1in,height=1.25in,clip,keepaspectratio]{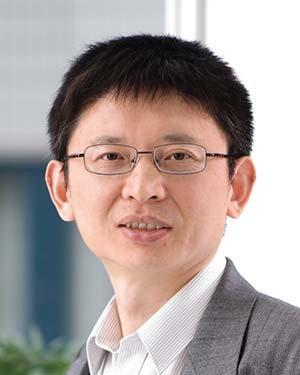}}]{Gang Feng} (Fellow, IEEE) received the B.Eng. and M.Eng. degrees in automatic control from Nanjing Aeronautical Institute, Nanjing, China, in 1982 and 1984, respectively, and the Ph.D. degree in electrical engineering from the University of Melbourne, Melbourne, VIC, Australia, in 1992.

From 1992 to 1999, he was a Lecturer/Senior Lecturer with the School of Electrical Engineering, University of New South Wales, Sydney, NSW, Australia. Since 2000, he has been with the City University of Hong Kong, Hong Kong, where he is a Chair Professor of Mechatronic Engineering. His current research interests include multi-agent systems and control, intelligent systems and control, and networked systems and control. 

Prof. Feng is the recipient of the IEEE Computational Intelligence Society Fuzzy Systems Pioneer Award, the \textsc {IEEE Transactions on Fuzzy Systems} Outstanding Paper Award, the Alexander von Humboldt Fellowship, the Changjiang Chair Professorship from the Education Ministry of China, the City University of Hong Kong Outstanding Research Award. He has been listed as an SCI Highly Cited Researcher by Clarivate Analytics since 2016. He is on the Advisory Board of \textit{Unmanned Systems}. He is/was an Associate Editor for \textsc {IEEE Transactions on Automatic Control}, \textsc {IEEE Transactions on Fuzzy Systems}, \textsc {IEEE Transactions on Systems, Man, and Cybernetics—Part C: Applications and Reviews}, \textit{Mechatronics}, \textit{Journal of Systems Science and Complexity}, and \textit{Journal of Control Theory and Applications}. 
\end{IEEEbiography}

\end{document}